\documentclass[%
preprintnumbers,
nofootinbib,
 amsmath,amssymb,
 aps,
 prd,
]{revtex4-2}

\usepackage{graphicx}% Include figure files
\usepackage{dcolumn}% Align table columns on decimal point
\usepackage{bm}% bold math
\usepackage{hyperref}% add hypertext capabilities
\usepackage{comment}
\usepackage{bbm}
\usepackage{tensor}
\usepackage{mathtools}
\usepackage[
	scr=esstix, % ESSTIX-thirteen 
	scrscaled=1]{mathalfa}
						
\usepackage{etoolbox}

\newcommand{\derivative}[2]{\frac{\mathrm{d}{#1}}{\mathrm{d}{#2}}}
\newcommand{\partialderivative}[2]{\frac{\partial{#1}}{\partial{#2}}}

\newcommand{\gradx}[1]{\mathrm{d}x^{#1}}
\newcommand{\grad}[1]{\mathrm{d}{#1}}

\newcommand{\complexnumbers}{\mathbb{C}}
\newcommand{\integernumbers}{\mathbb{Z}}
\newcommand{\realnumbers}{\mathbb{R}}
\newcommand{\unitelement}{\mathbbm{1}}
\newcommand{\unitelementof}[1]{\mathbbm{1}_{#1}}

\newcommand{\curlyA}{\mathcal{A}}
\newcommand{\curlyB}{\mathcal{B}}
\newcommand{\curlyH}{\mathcal{H}}

\newcommand{\curlyU}{\mathcal{U}}
\newcommand{\curlyh}{\mathscr{h}}
\newcommand{\curlyE}{\mathcal{E}}
\newcommand{\curlyF}{\mathcal{F}}
\newcommand{\curlyS}{\mathcal{S}}

\newcommand{\roundbrackets}[1]{\left({#1}\right)}
\newcommand{\bigroundbrackets}[1]{\bigl({#1}\bigr)}
\newcommand{\Bigroundbrackets}[1]{\Bigl({#1}\Bigr)}
\newcommand{\biggroundbrackets}[1]{\biggl({#1}\biggr)}

\newcommand{\curlybrackets}[1]{\left\{{#1}\right\}}

\newcommand{\bigcommutator}[2]{\bigl[{#1},{#2}\bigr]}
\newcommand{\biglangle}[1]{\bigl\langle{#1}\bigr\rangle}

\newcommand{\SUtwo}{\mathrm{SU}(2)}						% SU(2)
\newcommand{\Uone}{\mathrm{U}(1)}						% U(1)
\newcommand{\sutwo}{\mathfrak{su}(2)}					% su(2)
\newcommand{\G}{G}                                      % compact Lie group
\newcommand{\g}{\mathfrak{g}}                            % Lie algebra

\newcommand{\INT}[2]{\int\limits_{#1}^{#2}} 			% integral
\newcommand{\SUM}[2]{\sum\limits_{#1}^{#2}} 			% sum
\newcommand{\dx}[1]{\!\mathrm{d}{#1\,}}						% measure
\newcommand{\dxd}[2]{\!\mathrm{d}^{#1}{#2\,}}				% d dimensional measure
\newcommand{\INTset}[1]{{\int\!}_{#1}}

\newcommand{\DAGGER}{^\dagger}
\newcommand{\pr}{^\prime}								% simple prime
\newcommand{\AST}{^\ast}
\newcommand{\INV}{^{-1}}
\newcommand{\hodgestar}[1]{(\ast{#1})}

\newcommand{\commutator}[2]{\left[{#1},{#2}\right]}			% commutator
\newcommand{\innerproduct}[2]{\langle{#1,#2}\rangle}	% inner product

\newcommand{\innerproductof}[3]{\langle{#2,#3}\rangle_{#1}}

\newcommand{\diracdeltad}[2]{\delta^{(#1)}\left(#2\right)}				% Dirac delta generic d
\newcommand{\ket}[1]{\left|{#1}\right\rangle}				% Ket 
\newcommand{\set}[2]{\left\{\left.{#1}\,\right\vert\,{#2}\right\}}
\newcommand{\slotdot}{\boldsymbol{\,\cdot\,}}

\newcommand{\Ltwo}[1]{L^2(#1)}
\newcommand{\ltwo}{\ell^2(\mathbb{C})}
\newcommand{\norm}[1]{\lVert{#1}\rVert}
\newcommand{\absvalue}[1]{\lvert{#1}\rvert}

\newcommand{\scpr}[2]{\langle#1\, \vert \, #2 \rangle}
\newcommand{\sscpr}[3]{\langle#1\, \vert \, #2 \, \vert \, #3\rangle} 

\newcommand{\anni}{a}
\newcommand{\crea}{a^\dagger}

\newcommand{\keS}[2]{\kappa({#1},{#2})}

\newcommand{\Fhol}{h_e^{\mathcal{F}}}

\newcommand{\FflexpS}{V_S^{\mathcal{F}}}

\newcommand{\weylF}[2]{W^\mathcal{F}\bigl({#1},{#2}\bigr)}

\newcommand{\pair}[2]{({#1},{#2})}

\newcommand{\kernelS}[2]{{K_S}{#1}(#2)}

\newcommand{\regf}[1]{f_\epsilon(#1)}
\newcommand{\regflux}[2]{{X_\epsilon}{#1}(#2)}
\newcommand{\regsfield}[2]{\phi_\epsilon{#1}(#2)}
\newcommand{\regFflux}[2]{E_\epsilon{#1}(#2)}
\newcommand{\reggenEV}[2]{F_\epsilon{#1}(#2)}

\newcommand{\cfluxSf}{\mathcal{E}_S(f)} % classical flux for SU(2) setting
\newcommand{\cfluxSfpr}{\mathcal{E}_{S^\prime}(f^\prime)}

\newcommand{\FfluxSf}{X_{S,f}\otimes\mathbbm{1}_{\mathcal{F}}+\mathbbm{1}_\mathrm{AL}\otimes\left(\phi(\Gamma(S,f))+\cfluxSf\right)} % with condensate 
\newcommand{\FfluxSfpr}{X_{S^\prime,f^\prime}\otimes\mathbbm{1}_{\mathcal{F}}+\mathbbm{1}_\mathrm{AL}\otimes\left(\phi(\Gamma(S^\prime,f^\prime))+\cfluxSfpr\right)} % with condensate 

\newcommand{\anniF}[1]{a(#1)}
\newcommand{\creaF}[1]{a^\dagger(#1)}

\newcommand{\sfield}[1]{\phi(#1)}

\newcommand{\weylcf}[2]{W^\mathcal{F}\bigl({#1},{#2}\bigr)}

\newcommand{\eformfactor}[3]{{F_{#1}}^{#2}(#3)}
\newcommand{\Sformfactor}[3]{{F_{#1}}_{#2}(#3)}

\newcommand{\lplanck}{\ell_\mathrm{P}}
\newcommand{\rep}[1]{\pi\bigl(#1\bigr)}

\newcommand{\repAL}[1]{\pi_\mathrm{AL}\bigl(#1\bigr)}
\newcommand{\repopen}[2]{\pi_{#1}(#2)}

\newcommand{\repcf}[1]{\pi_{\mathcal{F}}\bigl(#1\bigr)}
\newcommand{\HScf}{\mathcal{H}_{\mathcal{F}}}

\newcommand{\HSFock}{\mathcal{H}_{\mathcal{F}}}
\newcommand{\repF}[1]{\pi_{\mathcal{F}}\bigl(#1\bigr)}

\newcommand{\symp}[2]{\mathbf{\sigma}\bigl({#1},{#2}\bigr)}

\newcommand{\cov}[2]{\mathbf{\alpha}\bigl({#1},{#2}\bigr)}

\newcommand{\covs}[2]{\mathbf{\alpha}\left({#1},{#2}\right)} % without _(S)
\newcommand{\covssym}{\mathbf{\alpha}}

\newcommand{\cons}[1]{\mathbf{\beta}\left({#1}\right)} % without _(S)
\newcommand{\conssym}{\mathbf{\beta}}

\newcommand{\CCR}[1]{\mathrm{CCR}({#1})}
\newcommand{\CCRW}[1]{\mathrm{CCR^{Weyl}}({#1})}

\newcommand{\state}[1]{\varphi\bigl({#1}\bigr)}
\newcommand{\statewrt}[2]{\varphi_{#1}\bigl({#2}\bigr)}
\newcommand{\aqfsstate}[1]{\varphi\bigl({#1}\bigr)}
\newcommand{\aqfss}{\varphi}

\newcommand{\cyl}{\mathrm{Cyl}}

\newcommand{\espu}[1]{(e^{#1},S^{#1})}
\newcommand{\espd}[1]{(e_{#1},S_{#1})}

\newcommand{\oisec}[2]{\mathrm{I}\left({#1},{#2}\right)}	

\newcommand{\vacexp}[1]{\bigl\langle{#1}\bigr\rangle_\Omega}
\newcommand{\vacexpof}[2]{\bigl\langle{#2}\bigr\rangle_{\Omega_{#1}}}

\newcommand{\ALHS}{\mathcal{H}_\mathrm{AL}}
\newcommand{\ALvac}{\Omega_\mathrm{AL}}

\newcommand{\vac}[1]{\Omega_{#1}}

\newcommand{\AL}{\mathrm{AL}}

\newcommand{\orientedintersection}[2]{\mathrm{I}\left({#1},{#2}\right)}				% oriented intersection number 
\newcommand{\cfluxS}{\mathcal{E}_{S}}

\newcommand{\flucS}{f_{S}}
\newcommand{\flucSpr}{f_{S^\prime}}

\newcommand{\weyl}[2]{W\bigl({#1},{#2}\bigr)}

\newcommand{\Cstar}{$C^\ast$}
\newcommand{\Star}{$\ast$}
\newcommand{\GNS}{(\pi,\mathcal{H},\Omega)}					% GNS triple
\newcommand{\GNSopen}[1]{(\pi_{#1},\mathcal{H}_{#1},\Omega_{#1})}	% GNS triple corresponding to 

\newcommand{\fluxSf}{E_S(f)}				%SU(2) flux with fixed surface S and function f
\newcommand{\fluxSfpr}{E_{S^\prime}(f^\prime)}
\newcommand{\fluxSfopen}[2]{E_{#1}(#2)}		%SU(2) flux with open slots
\newcommand{\fluxvfSf}{X_{S,f}}	
\newcommand{\fluxvfSfpr}{X_{S^\prime,f^\prime}}
\newcommand{\fluxvfSfopen}[2]{X_{#1,#2}}		%SU(2) flux vectorfield with open slots
\newcommand{\fluxvfacSf}[1]{X_{S,f}[#1]}				%SU(2) flux vectorfield action with fixed surface S and function f

\newcommand{\fluxS}{E_S}				%U(1) flux with fixed surface S 
\newcommand{\fluxvfS}{X_{S}}				%U(1) flux vectorfield with fixed surface S 
\newcommand{\holonomy}{h_e}				% holonomy along e
\newcommand{\holonomyopen}[1]{h_{#1}} % holonomy with open slot for path
\newcommand{\holonomyPOE}{\mathcal{P}\exp\bigl(-\INTset{e} A\bigr)}

\newcommand{\eg}{e.\,g.}						% e.g.
\newcommand{\ie}{i.\,e.}						% i.e.
\newcommand{\afree}{\mathfrak{A}_\text{free}}
\newcommand{\hfalgebra}{\mathfrak{A}_\text{HF}}
\newcommand{\flux}{\mathfrak{F}}

\newcommand{\cauchy}{\sigma}

\usepackage{amsthm}
\newtheorem{theorem}{Theorem}[section]  
\newtheorem{proposition}[theorem]{Proposition}  
\newtheorem{corollary}[theorem]{Corollary}
\newtheorem{lemma}[theorem]{Lemma}  
\newtheorem{definition}[theorem]{Definition} 
\newtheorem{conjecture}[theorem]{Conjecture}

\begin{document}

%\preprint{APS/123-QED}

\title{Towards Gaussian states for loop quantum gravity}% Force line breaks with \\
%\thanks{A footnote to the article title}%

\author{Hanno Sahlmann}
 \email{hanno.sahlmann@gravity.fau.de}
\author{Robert Seeger}%
 \email{robert.seeger@gravity.fau.de}
\affiliation{Friedrich-Alexander-Universit\"at Erlangen-N\"urnberg (FAU)\\ Institute for Quantum Gravity, Staudtstra{\ss}e 7/B2, 91058 Erlangen, Germany}

%\date{\today}% It is always \today, today,
             %  but any date may be explicitly specified

\begin{abstract}
An important challenge in loop quantum gravity is to find semiclassical states -- states that are as close to classical as quantum theory allows. This is difficult because the states in the Hilbert space used in loop quantum gravity are excitations over a vacuum in which geometry is highly degenerate. Additionally, fluctuations are distributed very unevenly between configuration and momentum variables. Coherent states that have been proposed to balance the uncertainties more evenly can, up to now, only do this for finitely many degrees of freedom. Our work is motivated by the desire to obtain Gaussian states that encompass all degrees of freedom. We reformulate the $\Uone$ holonomy-flux algebra in any dimension as a Weyl algebra. We then define and investigate a new class of states on this algebra which behave as quasifree states on the momentum variables. Using a general result on representations of the holonomy-flux algebra, we define analogous representations also in the case of non-Abelian compact structure groups. For the case of $\SUtwo$, we study an explicit example of such a representation and the consequences for quantum geometry. This kind of state, with Gaussian fluctuations in the spatial geometry, seems well suited to investigate problems related to structure formation in cosmology. 

%\begin{description}
%\item[Usage]
%Secondary publications and information retrieval purposes.
%\item[PACS numbers]
%May be entered using the \verb+\pacs{#1}+ command.
%\item[Structure]
%You may use the \texttt{description} environment to structure your abstract;
%use the optional argument of the \verb+\item+ command to give the category of each item. 
%\end{description}
\end{abstract}

%\keywords{Suggested keywords}%Use showkeys class option if keyword
                              %display desired
\maketitle

\tableofcontents

\section{\label{sec_introduction}Introduction}
In quantum field theory (QFT), the states can often be regarded as  excitations over a special state, such as a ground state or thermal state. This state encodes the physical circumstances such as the strength of the fluctuations or the temperature, and the excitations over it inherit many of its basic properties. In the following, we will loosely refer to such a special state as a \emph{vacuum}. 

Loop quantum gravity (LQG) is a QFT without classical background geometry.
The basic fields are a SU(2) connection one-form $A$ and a tensor density $E$  \cite{Ashtekar:1986yd, Barbero:1994ap}:
\begin{equation}\label{eq_ComAE}
    [A_a^i(x),E_j^b(y)]=i\hbar k \delta_a^b\delta_j^i\delta(x,y).   
\end{equation}
Here $k=8\pi G$, with Newton's constant $G$, and we have set the Barbero-Immirzi parameter to 1.  
The field algebra consists of parallel transporters and fluxlike variables 
\begin{equation}
\holonomy=\holonomyPOE, \qquad \fluxSf=\frac{1}{2}\INTset{S}f^j(x)E_j^a(x)\epsilon_{abc}\gradx{b}\wedge\gradx{c}
\end{equation}
where $e$ is a path and $S$ a surface in space. The commutation relations between these operators are purely topological in nature. The resulting algebra (along with its generalization to arbitrary dimension and structure group) is called holonomy-flux (HF) algebra.  Due to the nature of the commutation relations, the spatial diffeomorphisms act on the HF algebra as algebra automorphisms. The natural vacuum state, i.\,e., the one that is invariant under the action of the spatial diffeomorphisms, has very different properties from those in QFT on Minkowski space. Physically, this Ashtekar-Lewandowski (AL) state \cite{Ashtekar:1994mh} corresponds to a degenerate spatial metric $q_{ab}=0$ and a canonically conjugate extrinsic geometry with infinite fluctuations. This state is a natural ground state when general covariance is at the forefront. In fact, it is the unique diffeomorphism invariant state \cite{Fleischhack:2004jc, Lewandowski:2005jk} (see, however, \cite{Dziendzikowski:2009rv,Varadarajan:2007dk}).
The AL state is a special case in a whole family of states, all peaked on spatial geometry \cite{Koslowski:2007kh, Sahlmann:2010hn}. Apart from the AL state, these states are not invariant under spatial diffeomorphisms. 

At the other end of the spectrum, there is a construction of a vacuum state due to Dittrich and Geiller (albeit on a modified algebra) which is dual to the AL state in the sense that it is peaked on flat extrinsic geometry, while fluctuations in the spatial metric are maximal \cite{Dittrich:2014wpa,Bahr:2015bra}. See also \cite{Drobinski:2017kfm} for an elegant formulation of this idea in the Abelian case.  

When it comes to the description of classical spacetime, neither excitations above the AL vacuum nor above the Dittrich-Geiller vaccum are particularly suitable, due to the uneven distribution of fluctuations. AL excitations have been used to construct coherent states \cite{Thiemann:2000bw,Thiemann:2000ca,Thiemann:2000bx}, but these states have semiclassical properties only for finitely many degrees of freedom. Going over to infinitely many degrees of freedom leads to new measures on the space of connection fields \cite{Thiemann:2000by, Thiemann:2002vj}. 
It is also possible to transfer Gaussian measures from background dependent QFT to the space of connections used in LQG \cite{Varadarajan:1999it,Ashtekar:2001xp}, but the resulting Hilbert spaces so far do not support the holonomy flux algebra of LQG \cite{Sahlmann:2002xv}. 
Coming from a different angle, a series of works by Bianchi and collaborators explores states with high entanglement between neighboring regions of geometry \cite{Bianchi:2016tmw,Baytas:2018wjd,Bianchi:2018fmq}. Such states are also highly excited compared to the vacuum. As initially defined, they only comprise finitely many degrees of freedom. It is, however, possible to apply these ideas to Gaussian states on systems with infinitely many degrees of freedom \cite{Bianchi:2019pvv}.

The present work is also concerned with finding new states for LQG. For this it is necessary to first define the algebras carefully. We do this for the HF algebra in Sec. \ref{se:hfdef}. We should point out that we use a definition that is less strict than that of \cite{Lewandowski:2005jk,Stottmeister:2013qra}, in that it does not contain all relations among iterated commutators that are present in the AL representation. We also consider the case of the structure group $\Uone$ \cite{CorichiKrasnov98}.
This is interesting, because the relations of the HF algebra can be brought into the form of a Weyl algebra, by going over to the algebra elements \cite{Nekovar14}
\begin{equation}
\weyl{e}{S}=e^{\frac{i}{2}\orientedintersection{e}{S}}\holonomy e^{i\fluxS}.
\end{equation}
This formulation makes contact with free quantum fields on a fixed geometric background. 

As others before us, we are unable to find states for the HF algebra that are Gaussian with uncertainties split between the canonical variables. But we will describe a new type of state in which the flux operators have Gaussian fluctuations, whereas the properties of the holonomies are those of the AL representation. To be precise, the product of a holonomy and a flux has a vacuum expectation value 
\begin{equation}\label{eq:newstate}
    \langle 0 \lvert h_e\,  e^{iE_S(f)}\rvert 0 \rangle
    =\delta_{e,0}e^{-\frac{1}{2}\alpha_{S}(f,f)}e^{i\mathbf{\beta}_S(f)}, 
\end{equation}
where $\alpha_{S}$ is an $S$-dependent bilinear form and $\beta_S$ a linear one. 
This kind of state gives a new representation of the HF algebra, in which the spatial geometry fluctuates around an average value given by $\mathbf{\beta}_S$. If one chooses the covariance $\alpha$ to be vanishing, one obtains states of the type considered in \cite{Koslowski:2007kh, Sahlmann:2010hn,Varadarajan:2013lga,Campiglia:2013nva,Campiglia:2014hoa}. 

One motivation for the consideration of these states is the quantum origin of the primordial perturbations. The current observations of the cosmic microwave background (CMB) can be described by saying that primordial perturbations of the density and the spatial metric seem to be described by a Gaussian random field with a certain covariance. A state of the form \eqref{eq:newstate} describes a quantized spatial geometry that fluctuates around a background given by $\mathbf{\beta}$, with Gaussian correlations given by $\alpha$. Thus the states that we describe might be well suited to describe the quantum geometry of the early universe. 

We will start the discussion by reviewing some aspects of LQG  and the definition and properties of Weyl algebras, in Secs. \ref{se:lqg} and \ref{se:weyl}, respectively. 
We continue with a precise definition of the HF algebra for arbitrary dimension and compact gauge group in Sec. \ref{se:hfdef}, and we state and prove a result on modifications of its representations in Sec. \ref{sec_RepsExtFluxes}.
In Sec. \ref{sec_WeylAlgUone} we will discuss the $\Uone$ model in detail and show that its HF algebra derives from a Weyl algebra. New states for this algebra are then defined in Sec. \ref{se:aqfree}. 
Analogous states in the case of a non-Abelian gauge group are discussed in Sec. \ref{se:sutwo}, with the  definition of the new states in Sec. \ref{se:sutwoqfree}, and a discussion of the changes of the area spectrum due to the fluctuations in Sec. \ref{se:area}. We end with a short summary and outlook, Sec. \ref{se:outlook}. 

\subsection{Quantization of diffeomorphism invariant theories of connections}\label{se:lqg}
We consider the canonical formulation of a gauge theory on a globally hyperbolic ($D+1$)-dimensional spacetime with a Cauchy surface $\cauchy$. We assume the structure group $\G$ to be a compact Lie group. We chose the canonical position variable to be a $\G$ connection $A$ on $\cauchy$. The canonical momentum is then a densitized vector field $E$ taking values in $\g$. Ashtekar-Barbero variables in LQG \cite{Ashtekar:2004eh,Thiemann2007}
are a special case of this, with $\G=\SUtwo$ and $D=3$. 

In a setting without a fixed spacetime metric, natural variables beyond the point fields $A$ and $E$ are holonomies of $A$ along one-dimensional paths and integrals of $E$ over ($D-1$)-dimensional surfaces, an electrical flux. We write
\begin{equation}
\holonomy=\holonomyPOE
\end{equation}
for the  $\G$ valued holonomy along $e$.  For the fluxes of $E$ we smear against $\g$-valued functions $f$ with support on a ($D-1$)-dimensional, orientable hypersurface $S$, given by
\begin{equation}\label{eq_flux}
\fluxSf=\frac{1}{(D-1)!}\INTset{S}f^j(x)E_j^a(x)\epsilon_{ab_1\ldots b_{D-1}}\gradx{b_1}\wedge\ldots \wedge\gradx{b_{D-1}}.
\end{equation}
The canonical commutation relation of holonomies and fluxes is only sensitive to intersection points of the corresponding edge and surface. For a pair of edge and surface intersecting in only one of the end points of $e$ it is given by
\begin{equation}\label{eq_LQG_HolFluxCommRel}
\commutator{\holonomy}{\fluxSf}=\frac{\hbar k}{4i}\,\kappa(e,S)
\begin{cases}
\holonomy f(p) & \text{ for } p=e\cap S \text{~source of~}e\\
- f(p)\holonomy & \text{ for } p=e\cap S \text{~target of~}e 
\end{cases}.
\end{equation}
Here, $\kappa(e,S)$ encodes partly the relative orientation of edge and surface. It is $\pm1$ for an edge above or below the surface.  

Sums and products of holonomy functionals span the so-called cylindrical functions. The definition is as follows. Consider a graph embedded in $\cauchy$, i.e. a finite collection of edges that are allowed to build out vertices by intersections of their beginning and final points, {\ie} $\gamma=\curlybrackets{e_1,e_2,\dots,e_n}$. One further speaks of the set of edges $E(\gamma)$ and the set of vertices $V(\gamma)$ of a given graph $\gamma$. 
Now we call $\curlyA$ the space of smooth connections and look at functions on this space,
\begin{equation}\label{eq_CylFuncMap}
F:\curlyA\rightarrow\complexnumbers.
\end{equation}
Such a function is said to be smooth cylindrical with respect to a given graph $\gamma$ if there is a smooth function 
\begin{equation}
F_\gamma:{\G}^{\absvalue{E(\gamma)}}\rightarrow\complexnumbers,
\end{equation}
such that the function on $\curlyA$ can be expressed by this function of powers of $\G$, {\ie} by setting for $A\in\curlyA$
\begin{equation}
F(A)=F_\gamma\roundbrackets{\curlybrackets{\holonomy(A)}_{e\in{E(\gamma)}}}. 
\end{equation}
The smooth cylindrical functions form an algebra which we will call $\cyl$.\footnote{Note that the smooth cylindrical functions are sometimes referred to as $\cyl^\infty$, in contrast to our present notation which is in keeping with \cite{Lewandowski:2005jk}.} 
The fluxes $\fluxSf$ can be used to define the Hamiltonian vector fields  
\begin{equation}
\fluxvfSf=\commutator{\fluxSf}{\slotdot},
\end{equation} 
which act on $\cyl$.  

The \Star-Lie algebra on which the quantum theory may be based is the algebra generated by smooth cylindrical functions and the flux vector fields by commutators as multiplication, subject to complex conjugation as involution. The resulting quantum \Star-algebra is known as the HF algebra. For the precise definition that we will use, see Sec. \ref{se:hfdef}. 

There is a unique diffeomorphism invariant representation of the HF algebra \cite{Lewandowski:2005jk,Fleischhack:2004jc}. 
\begin{theorem}
For $D\geq 2$ the Ashtekar-Isham-Lewandowski state 
\begin{equation}
\statewrt{\mathrm{AL}}{F\fluxvfSfopen{S_1}{f_1}\fluxvfSfopen{S_2}{f_2}\dots\fluxvfSfopen{S_n}{f_n}}=
\begin{cases}
0 & n>0\\
\mu_0(F) & n=0
\end{cases}
\end{equation}
is the only diffeomorphism invariant state on the HF algebra. Here, the state acts on elements of $\cyl^\infty$ as 
\begin{equation}
\mu_0(F)= \INTset{\G^{\absvalue{E(\gamma)}}}\prod\limits_{e\in E(\gamma)}\dx{\mu_\mathrm{H}}(g_e) F_\gamma\roundbrackets{\curlybrackets{g_e}_{e\in E(\gamma)}}.
\end{equation} 
\end{theorem}
A sequence of holonomies and flux vector fields can always be brought into a normal ordered form --~flux vector fields to the right and holonomies to the left~-- by using the commutation relations. The resulting expression is by linearity of the states a sum of terms of the form used in the theorem. Furthermore, the measure that is used in the integration over the powers of $\G$ is just the product of the Haar measures for the individual copies of $\G$.   

The representation that arises from this state via the Gelfand-Naimark-Segal (GNS) construction is called AL representation \cite{Ashtekar:1994mh,Ashtekar:1994wa}. The Hilbert space is the space of square integrable functions over the space of generalized --~specifically distributional~-- connections:
\begin{equation}
\curlyH=\mathrm{L}^2\roundbrackets{\overline{\curlyA},\dx{\mu_\mathrm{AL}}}.
\end{equation}

The representation of holonomies and cylindrical functions and fluxes is similar to the ordinary Schr\"odinger representation, {\ie} as a multiplication operator and a derivative. For the cylindrical functions one sets 
\begin{equation}\label{eq:ALh}
\bigroundbrackets{\repopen{\mathrm{AL}}{F}\Psi}(A)=F(A)\Psi(A).
\end{equation}

The AL representation of the fluxes that act on cylindrical functions is the following:
\begin{equation}\label{eq:ALE}
\bigroundbrackets{\repopen{\mathrm{AL}}{\fluxSf}\Psi}(A)=\fluxvfacSf{\Psi}=\frac{\hbar k}{2}\SUM{v\in V(\gamma)}{}f^j(v)\SUM{e\in E(v)}{}\kappa(e,S)J^{(v,e)}_j\Psi(A).
\end{equation} 
The first sum runs over all vertices of the graph $\gamma$ underlying the cylindrical function $\Psi$. Here it is assumed that every intersection point between edges and the surface is a --~possibly trivial~-- vertex considered in $V(\gamma)$.  At these vertices one evaluates the smearing function of the flux vector fields. The second sum now runs over all edges $e$ that begin or end at the vertex $v$. The object $J^{(v,e)}_j$ encodes the left- and right-invariant vector fields acting on copies of $\G$ assigned to the specific edges of a vertex, for out- and ingoing edges, respectively \cite{Lewandowski:2005jk}.

This algebra arose from the consideration of the special case $\G=\SUtwo$ in LQG. It has the advantage that diffeomorphisms of $\cauchy$ act on the algebra in a very simple way. 

\subsection{Review of Weyl algebras}\label{se:weyl}
In the following we want to give a short introduction to the topic of Weyl algebras of canonical commutation relations and quasifree states. We base our discussion on \cite{Petz1990,Derezinski2013}.

A Weyl algebra is a \Cstar-algebra, solely constructed from a canonical commutation relation (CCR). It is sensible to present the construction in several steps. 

\begin{definition}
\label{def_CCRstarAlgebra}
The CCR \Star-algebra over the presymplectic space $(H,\sigma)$, denoted by $\CCR{H,\sigma}$, is the algebra generated by the Weyl elements $W(X)$, $X\in H$ that satisfy the Weyl relations 
\begin{equation}
W(X)W(Y)=e^{-\frac{i}{2}\symp{X}{Y}}W(X+Y)
\end{equation}
and are subject to the involution $\ast$, such that 
\begin{equation}
W(X)\AST=W(-X)=W(X)^{-1}.
\end{equation}
\end{definition}

The Weyl elements are unitary with respect to the involution, and hence a representation of this algebra has to be a unitary representation.

\begin{definition}
\label{def_CCRrep}
Let $\curlyH$ be a Hilbert space and let $\pi:H\rightarrow\curlyU(\curlyH)\subset\curlyB(\curlyH)$, $X\mapsto\pi(W(X))$ be a map from the real vector space $H$ into the unitary operators on $\curlyH$. Then $(\pi,\curlyH)$ is a representation of a canonical commutation relation in terms of a  presymplectic form or equivalently a CCR representation over the presymplectic space $(H,\sigma)$, if the so-called Weyl operators satisfy the Weyl relations 
\begin{equation}
\rep{W(X)}\rep{W(Y)}=e^{-\frac{i}{2}\symp{X}{Y}}\rep{W(X+Y)}
\end{equation}
and $\pi$ is a $\ast$-homomorphism, {\ie} $\rep{W(X)\AST}=\rep{W(X)}\DAGGER$.
\end{definition}

In principle, $\CCR{H,\sigma}$ is homomorphic to a \Cstar-algebra of bounded operators on a given representation Hilbert space, because of the unitarity of the Weyl elements. In order to promote the \Star-algebra to a \Cstar-algebra one exploits the fact that with $(\pi_{\ell^2},\ell^2(H))$ there is always a \Star-representation, which allows for the definition of a norm that is highly dependent on CCR \Star-representations. It is known as the minimal regular norm \cite{Petz1990}. 

\begin{definition}
\label{def_WeylAlgebra}
We equip the \Star-algebra $\CCR{H,\sigma}$ with the minimal regular norm 
\begin{equation}
\label{eq_minimal_regular_norm}
\norm{A}:=\sup\set{\norm{\rep{A}}_\curlyH}{(\pi,\curlyH)\text{~is a CCR \Star-representation}}.
\end{equation}
The \Cstar-algebra, obtained by the completion of $\CCR{H,\sigma}$ with respect to the minimal regular norm, is referred to as Weyl algebra or Weyl CCR algebra over the presymplectic space $(H,\sigma)$:
\begin{equation}
\label{eq_weyldef}
\CCRW{H,\sigma}:=\overline{\CCR{H,\sigma}}.
\end{equation}
\end{definition}
Equation \eqref{eq_minimal_regular_norm} is meaningful since there is always a representation of $\CCR{H,\sigma}$ by explicit construction on the Hilbert space of square summable sequences indexed by $H$. In the case of a nondegenerate $\sigma$, the corresponding Weyl algebra is also unique in the sense that any \Cstar-algebra generated by elements $W(X)$ as in definition \ref{def_CCRrep} is isomorphic to $\CCRW{H,\sigma}$ \cite{Slawny:1972iq}. On $\CCR{H,\sigma}$ there exists a tracial state $\tau$ with the property \cite{Slawny:1972iq,Petz1990}
\begin{equation}
\tau\left(\sum_{X\in H}f(X)W(X)\right)= f(0), 
\end{equation}
where it is understood that the function $f:H\rightarrow\complexnumbers$ is nonzero only in finitely many places. 

It is useful to note for later that there cannot be linear relations among the $\weyl{e}{S}$. 
\begin{lemma}
\label{le:linear_independence}
The $W(X), X\in H$ form a basis of $\CCRW{H,\sigma}$. 
\end{lemma}
\begin{proof}
Consider any linear combination $a=\sum c_i W(X_i)$ of mutually distinct elements $W(X_i)$, where we assume $c_i\neq 0\;\forall\; i$ . Using the state $\tau$ we can estimate 
\begin{equation}
   \norm{a}^2\geq \tau(a^*a) = \sum_i |c_i|^2,   
\end{equation}
and so $a$ can never be 0. 
\end{proof}
The notion of quasifree states has been introduced by Araki in the 1960s \cite{Araki63,Araki:1970zza,Araki:1970zz}. This is the most general class of Gaussian states, {\ie} expectation values of Weyl elements are Gaussian functions, or $n$-point correlation functions of field operators separate into partitions of two-point correlation functions.

\begin{definition}
Let $\varphi$ be a state on $\CCRW{H,\sigma}$. It is called quasifree if it is of the form
\begin{equation}
\state{W(X)}=e^{-\frac{1}{2	}\cov{X}{X}}\qquad\forall X\in H,
\end{equation}
where $\alpha:H\times H\rightarrow \realnumbers$ is a positive symmetric  bilinear form on $H$, called covariance of the quasifree state.
\end{definition}

The existence of such quasifree states is tied to a condition on the covariance and the symplectic form, namely 
\begin{equation}\label{eq_inequality}
\sqrt{\cov{X}{X}}\sqrt{\cov{Y}{Y}}\geq \frac{1}{2}\absvalue{\symp{X}{Y}}\qquad \forall X,Y \in H.
\end{equation}
With this, the existence is reduced to finding a covariance that satisfies the determining inequality. 

We want to list some important features of quasifree states and refer to \cite{Petz1990,Derezinski2013} for details. Let $\GNS$ be the GNS representation of the quasifree state. 
First of all, quasifree states are regular.
Considering this, Stone's theorem applies and from the representation of the Weyl  elements one obtains self-adjoint field operators:
\begin{equation}
B_\pi(X)=\frac{1}{i}\derivative{\phantom{t}}{t}\rep{W(tX)}\big|_{t=0}=\frac{1}{i}\derivative{\phantom{t}}{t} e^{iB_\pi(X)t}\big|_{t=0}.
\end{equation}
These field operators are linear in their arguments and satisfy the following commutation relations on $\mathcal{D}=\operatorname{span} \{W(X)\Omega \,|\, X\in H \}$:
\begin{equation}\label{eq_commutator_Weyl_field}
\begin{aligned}
&\commutator{B_\pi(X)}{B_\pi(Y)}=i\symp{X}{Y},\\
&\commutator{\rep{W(X)}}{B_\pi(Y)}=-\symp{X}{Y}\rep{W(X)}.
\end{aligned}
\end{equation}
The two-point correlation functions only depend on the covariance and the symplectic form:
\begin{equation}
\biglangle{B_\pi(X)B_\pi(Y)}_\Omega=\cov{X}{Y}+\frac{i}{2}\symp{X}{Y}.
\end{equation} 

The separation of correlation functions into two-point functions follows from Stone's theorem since
\begin{equation}
\begin{aligned}
&\bigl\langle{B_\pi(X_1)B_\pi(X_2)\dots B_\pi(X_n)}\bigr\rangle_\Omega=\frac{1}{i^n}\partialderivative{^n}{t_1\partial t_2\dots\partial t_n}\state{W(t_1X_1)W(t_2X_2)\dots W(t_nX_n)}\big|_{t=0}\\
&=\frac{1}{i^n}\partialderivative{^n}{t_1\partial t_2\dots\partial t_n} 
\exp\curlybrackets{-\frac{1}{2} \SUM{j=1}{n} (t_j)^2 \cov{X_j}{X_j}-\SUM{j<k}{}t_jt_k\roundbrackets{\cov{X_j}{X_k}+\frac{i}{2}\symp{X_j}{X_k}}}
\Bigg|_{t=0}.
\end{aligned}
\end{equation} 
Clearly this vanishes for odd $n$. For $n$ even in contrast, there is still only one nonvanishing term, which carries a linear contribution for every single $t_k$. This is exactly the term that is separated into a two-point function. Summarizing, we have
\begin{equation}\label{eq_CCR_WickDecompFieldOperators}
\bigl\langle B_\pi(X_1)B_\pi(X_2)\dots B_\pi(X_n)\bigr\rangle_\Omega=
\left\{
\begin{array}{l l}
0  & n \text{~odd}\\
\underset{\text{partitions}}{\sum\prod} \bigl\langle B_\pi(X_i)B_\pi(X_j)\bigr\rangle_\Omega & n \text{~even}
\end{array}
\right. .
\end{equation}

\section{A general class of representations of the holonomy-flux algebra}\label{sec_generalresult}
To properly talk about different representations, we have to precisely define the algebra that we have informally described in Sec. \ref{se:lqg}. After giving this definition, we will state and prove a general result about its representations that will underlie the construction of the new states in later parts of this work. 

\subsection{The holonomy-flux algebra}\label{se:hfdef}
A precise definition of the HF algebra for general $D$ and $\G$ as a quotient of a free algebra can be found in \cite{Lewandowski:2005jk}, and a similar one for $D=3$ in \cite{Stottmeister:2013qra}. Both definitions take into account most of the relations that are present in the AL representation of holonomies and fluxes. In particular, the fluxes generate an infinite dimensional Lie algebra with additional relations among its elements. For example, it was pointed out in \cite{Stottmeister:2013qra} that in the case $D=3$ the remarkable identity
\begin{equation} \label{eq:triplecomm}
[\fluxSf,[\fluxSfopen{S}{f'},\fluxSfopen{S}{f''}]]= \frac{1}{4} \fluxSfopen{S}{[f,[f',f'']]}
\end{equation}
holds on certain states in the AL representation, and by extension in the Lie algebra generated by the fluxes. 

On the other hand, it was noted in \cite{Lewandowski:2005jk} that only a few of the relations are really necessary to show uniqueness of diffeomorphism-invariant representations. More generally, only a small subset of relations of a given classical (Poisson) algebra can be respected by the commutation relations in a quantum theory of the given system, as was first pointed out by Groenewold and Van Hove. Thus one can argue (see for example \cite{Sahlmann:2010hn}) that only the relations that are absolutely characteristic for the LQG approach should be taken over to the quantum algebra. This would mean requiring the basic commutator 
\begin{equation}\label{eq:fluxf}
    [\fluxSf,F]=\fluxvfSf[F]
\end{equation}
as a relation in the quantum algebra, but not a more complicated relation such as \eqref{eq:triplecomm}. 
Here, we will follow \cite{Sahlmann:2010hn} in taking fewer relations into account. Besides linear and adjointness relations, we require only the multiplicative structure among cylindrical functions and \eqref{eq:fluxf}. These relations imply further relations via the Jacobi identities. We will comment on this below. 

We will mathematically define the algebra by starting with a free algebra over a certain set of symbols and then dividing by the two-sided ideal generated by the relations listed above. This has the advantage that we automatically impose all the relations that are implied by those that we list explicitly.  

Let us first define the label set for the quantum algebra $\hfalgebra$.  On the one hand, we use the cylindrical functions $\cyl$ as labels. On the other hand,  consider elementary flux variables on phase space, i.e. functionals of the form $\fluxSf$ for a surface $S$ and a smearing field $f$. In particular, these variables satisfy
\begin{equation}
\label{eq_fluxlables}
    \fluxSfopen{S}{f+f'}=\fluxSf+\fluxSfopen{S}{f'}, \qquad \fluxSfopen{-S}{f}=\fluxSfopen{S}{-f},
\end{equation}
where $-S$ is the surface $S$ as a manifold, but with the opposite orientation.
Let $\flux$ be the vector space of phase space functions that are finite linear combinations of elementary flux variables $\fluxSf$. In particular, this vector space incorporates the relations \eqref{eq_fluxlables}.

Now we are in a position to define the free algebra we will be starting from. Take algebra $\afree$ of formal linear combinations of sequences of elements of $\cyl \cup \flux$, with multiplication and $*$ defined by
\begin{align}
    (a_1,a_2,\ldots, a_m)\cdot(b_1,b_2,\ldots,b_n)&:= (a_1,a_2,\ldots, a_m, b_1,b_2,\ldots,b_n),\\
    (a_1,a_2,\ldots, a_m)^*&:= (\overline{a_m},\overline{a_{m-1}},\ldots, \overline{a_1}),
\end{align}
and linear extension. Here all entries $a_1,b_1,\ldots$ are from $\cyl \cup \flux$ and $\overline{a}$ is complex conjugation of functions on phase space. 
The algebra $\hfalgebra$ that we will be working with is the algebra of equivalence classes under the relations defined by a certain ideal in $\afree$. Let 
\begin{equation}\label{eq:quotient}
    \hfalgebra= \afree/I
\end{equation}
where $I$ is the two-sided ideal generated by the following classes of elements: 
\begin{align}
     &\alpha (a)- (\alpha a),\label{eq:g1}\\
     &(a)+(b)-(a+b),\label{eq:g2}\\
     &(F_1,F_2)-(F_1F_2),\label{eq:g3}\\
     &(\fluxSf,F)-(F,\fluxSf) - (\fluxvfSf[F]).\label{eq:g4}
\end{align}
Here $\alpha\in \complexnumbers$, $a,b\in \cyl \cup \flux$, and  $F_1,F_2,F\in\cyl$. This means that elements of $\hfalgebra$ are equivalence classes of elements in $\afree$, with the product
\begin{equation}\label{eq_ProductEqivClass}
    [x][y]=[xy], \quad x,y\in \afree, 
\end{equation}
which is well-defined because $I$ is a two-sided ideal. $I$ is closed under the $*$-operation in $\afree$, so the $*$-operation defined by
\begin{equation}
    [x]^*:=[x^*], \quad x\in \afree 
\end{equation}
is well-defined on $\hfalgebra$. 

How can we define representations of $\hfalgebra$? A practical way to do this is to first specify a representation $\widetilde{\pi}$ of $\afree$, by specifying $\widetilde{\pi}$ on a basis of $\cyl$ and one of $\flux$. Then we have the following result.
\begin{lemma}\label{le:repa}
A representation $\widetilde{\pi}$ of $\afree$ defines a representation of $\hfalgebra$ by
\begin{equation}\label{eq:pidef}
    \pi([x]):=\widetilde{\pi}(x)
\end{equation}
if the generators \eqref{eq:g1}--\eqref{eq:g4} of the ideal $I$ are in the kernel of $\pi$. 
\end{lemma}
This is an immediate consequence of the fundamental theorem on algebra homomorphisms, and we omit the proof.

In \cite{Stottmeister:2013qra} some doubts regarding the consistency of the results of \cite{Sahlmann:2010hn} were expressed. These doubts also have relevance for the contents of the current work, so we would like to explicitly address them. The doubts have to do with the relations implied by those generating the ideal $I$. Indeed, it is not easy to enumerate the relations contained in $I$, and there are many. In particular, there are all the relations following from Jacobi identities for commutators, since the Jacobi identities are implied by the definition
\begin{equation}
    [[a],[b]]:= [a][b]-[b][a], \qquad a,b\in \afree.
\end{equation}
One can for example check that, as was pointed out in \cite{Stottmeister:2013qra}, for any $F\in \cyl$
\begin{equation}\label{eq_tripleCommutatorEF}
    [[\fluxSf,[\fluxSfopen{S}{f'},\fluxSfopen{S}{f''}]],F]= \frac{1}{4} [\fluxSfopen{S}{[f,[f',f'']]},F]
\end{equation}
in $\hfalgebra$, where here and in the following we drop the brackets $[\,\cdot\,]$ denoting equivalence classes for better readability.   
Note that in contrast the relation \eqref{eq:triplecomm} does \emph{not} hold in $\hfalgebra$ as defined above. 

We also note that there cannot be any inconsistency in the relations. The worst that can happen is that (i) $I=\afree$, and hence $\hfalgebra$ becomes trivial, or (ii) there does not exist nontrivial representations $\widetilde{\pi}$ that have the generators \eqref{eq:g1}--\eqref{eq:g4} of $I$ in the kernel. The existence of the AL representation \eqref{eq:ALh} and \eqref{eq:ALE} shows, however, that (i) does not happen, as the operators in that representation fulfill the relations \eqref{eq:g1}--\eqref{eq:g4} and are not trivial. Indeed, because of this the AL representation is a representation of $\hfalgebra$ in the literal sense, and hence also (ii) is not a problem. 

The representations of \cite{Sahlmann:2010hn} fulfill the assumptions of Lemma \ref{le:repa}. Therefore they are representations of $\hfalgebra$ as defined above. The same holds for the class of representations that we will propose below. One problem with \cite{Sahlmann:2010hn} is that the construction of the quotient algebra \eqref{eq:quotient} is invoked without giving much detail and hence inviting misunderstandings. 

\subsection{Representations shifting the fluxes}\label{sec_RepsExtFluxes}
In this section, we will make an observation regarding the representations of the holonomy-flux algebra $\hfalgebra$ as defined above: Since we only fix the commutators between holonomies and fluxes (via the ideal \eqref{eq:g4}), and not between the fluxes themselves, there is more freedom for defining representations. In particular, this allows for a shift in the representatives of the fluxes. The detailed statement is as follows. 
\begin{proposition}
\label{prop_flux_shift}
Let $\pi$ be a representation of $\hfalgebra$ on $\mathcal{H}$. Let
\begin{equation}
    \Xi:\mathfrak{F}\longrightarrow L(\mathcal{K})
\end{equation}
be a linear map from the set of linear combinations of elementary flux variables to operators on a Hilbert space $\mathcal{K}$. Then there is a new representation $\pi'$ on $\mathcal{H}\otimes\mathcal{K}$ in which
\begin{equation}
    \pi'(F)=\pi(F), \qquad \pi'(E_S(f)) = \pi(E_S(f))+\Xi(E_S(f))
\end{equation}
\end{proposition}{}
\begin{proof}
Assume $\pi$, $\Xi$ given as in the proposition. Note first that the converse of Lemma \ref{le:repa} holds true as well, in the sense that we can extend $\pi$ to a representation $\widetilde{\pi}$ of $\afree$ by declaring 
\begin{equation}
   \widetilde{\pi}(F):= \pi(F), \qquad \widetilde{\pi}(\fluxSfopen{S}{f}):= \pi(\fluxSfopen{S}{f})
\end{equation}
and representing sequences by products of these elementary ones. Then we can set 
\begin{equation}
   \widetilde{\pi}'(F):= \pi(F), \qquad \widetilde{\pi}'(\fluxSfopen{S}{f}):= \pi(\fluxSfopen{S}{f})\otimes \unitelement + \unitelement\otimes \Xi(\fluxSfopen{S}{f}).
\end{equation}
and extend to all of $\afree$ by representing sequences by products of the elementary representatives. This is a representation of $\afree$, since there are no nontrivial relations to fulfill.
Now we have to check that the ideals \eqref{eq:g1}--\eqref{eq:g4} are in the kernel of $\widetilde{\pi}'$. Equations \eqref{eq:g1} and \eqref{eq:g2} are in the kernel due to the linearity of $\Xi$. Equation \eqref{eq:g3} is in the kernel by construction. Only the commutation relations \eqref{eq:g4} require an explicit check:
\begin{align*}
    \widetilde{\pi}' \Big( &(\fluxSf,F)-(F,\fluxSf) - (\fluxvfSf[F]) \Big)& \\ &=\Bigl(\widetilde{\pi}(\fluxSf)\otimes\unitelement+\unitelement\otimes\Xi(\fluxSf)\Bigr)\widetilde{\pi}(F)\otimes\unitelement
    - \widetilde{\pi}(F)\otimes \unitelement\Bigl(\widetilde{\pi}(\fluxSf)\otimes\unitelement+\unitelement\otimes\Xi(\fluxSf)\Bigr) 
    - \widetilde{\pi}(\fluxvfSf[F])\otimes\unitelement\\
    &= \Bigl([\widetilde{\pi}(\fluxSf), \widetilde{\pi}(F)]-\widetilde{\pi}(\fluxvfSf[F]\Bigr) \otimes\unitelement + \widetilde{\pi}(F)\otimes\Xi(\fluxSf)-\widetilde{\pi}(F)\otimes\Xi(\fluxSf)\\
    &=0,
\end{align*}
where we have used that $\widetilde{\pi}$ comes from a representation of $\hfalgebra$ and hence annihilates the ideals. 
\end{proof}
We will use this freedom to define representations that mirror properties of quasifree representations in the following sections.  

\section{Relation between Weyl algebras and the holonomy-flux algebra for the structure group $\Uone$}\label{sec_WeylAlgUone}

The electromagnetic analog of loop quantum gravity is the situation where one considers the kinematics of a $\Uone$ Yang-Mills theory and quantizes it in the spirit of diffeomorphism invariant quantum gravity \cite{CorichiKrasnov98,Thiemann:1997rt}. We want to discuss the $\Uone$ case in general and look at a formulation in $D+1$ dimensions and hence a $D$-dimensional spatial manifold.

We consider a pair consisting of a $\Uone$ connection or vector potential $A_a(x)$ and an electric field $E^a(x)$, set coupling constants, including $\hbar$, to 1 and stipulate the following CCR:
\begin{equation}
\begin{aligned}\label{eq_WeylUone_CommutatorConnectionElectricfield}
&\commutator{A_a(x)}{E^b(y)}=i\delta^b_a\diracdeltad{D}{x,y},\\
&\commutator{A_a(x)}{A_b(y)\vphantom{E^b(y)}}=0=\commutator{E^a(x)}{E^b(y)}.
\end{aligned}
\end{equation}
The smeared connection is defined completely analogous to the $\SUtwo$ case. For the smeared flux, however, we go without an additional smearing function and consider only the integration over the surface.\footnote{This is possible because the $E$ field is gauge invariant and hence globally defined. Fluxes with smearing functions can be approximated by suitable linear combinations of fluxes without, so this set of functions also seems big enough. We chose it to simplify the formalism, but expect no difficulty in extending the results to the bigger space of fluxes.} Smeared connection and flux therefore are given by 
\begin{equation}
\begin{aligned}
A(e)&=\INTset{e}A=\INT{0}{1}\dx{t}A_a(e(t))\dot{e}^a(t),\\
E(S)&=\INTset{S}\hodgestar{E}=\INTset{S}\frac{1}{(D-1)!}E^a\epsilon_{aa_1a_2\dots a_{D-1}}\gradx{a_1}\wedge\gradx{a_2}\wedge\dots\wedge\gradx{a_{D-1}}.
\end{aligned}
\end{equation}  
The commutation relation of these variables is characterized by the oriented intersection number $\oisec{e}{S}$ of a surface $S$ and a path $e$. This object has already been looked at in the context of LQG in \cite{Ashtekar:1996bs,Ashtekar:1997rg} under the name Gauss linking number. For the CCR, we find
\begin{equation}
\commutator{A(e)}{E(S)}=i\orientedintersection{e}{S}.
\end{equation}
The actual commutation relation that is analogous to $\SUtwo$ considers holonomies.
For $\Uone$ there is no need for path ordering in the holonomies and we can simply look at the object 
\begin{equation}
\holonomy=e^{iA(e)}.
\end{equation}
As in the non-Abelian theory, the holonomies are elements of the structure group and for that reason elements of $\Uone$. 
The commutation relation for holonomies and fluxes becomes
\begin{equation}\label{eq_ComRelUoneHF}
\commutator{\holonomy}{E(S)}=-\orientedintersection{e}{S}\holonomy.
\end{equation}
This relation fully characterizes the $\Uone$ holonomy-flux algebra. Fluxes and holonomies commute, respectively.

In the following sections, we will not consider a strict specialization of the general result of section \ref{sec_generalresult}, as the Abelian structure group $\Uone$ allows for a richer structure. We will show that it is possible to construct a Weyl algebra $\CCRW{H,\sigma}$ for $\Uone$, which will give us representations of the $\Uone$ holonomy-flux algebra. The precise statement that we will establish is as follows.

\begin{proposition}\label{prop_UoneWeylRepHF}
Every representation of $\CCRW{H,\sigma}$ that is regular in the fluxes gives a representation of the $\Uone$ holonomy-flux algebra. 
\end{proposition}

Because of this relation to the HF algebra, we will refer to $\CCRW{H,\sigma}$ also as the HF Weyl algebra. 

\subsection{Distributional form factors for edges and surfaces}
\label{sec_distributionalformfactors}
For the Weyl algebra, we need at least a presymplectic space and hence an explicit vector space. This will not be an ordinary vector space but a certain space of edges and surfaces. The formulation that turns out to be the most suitable for such a vector space is in terms of distributional objects, which we want to call form factors, describing edges and surfaces. 
A similar notion for edges has already been used in \cite{Ashtekar:1996bs,Ashtekar:1997rg}, although the treatment of surfaces is somewhat different, since they are described by their closed boundary curves. 

\begin{definition}
\label{def_formfactor}
\begin{enumerate}\item[]
\item Let $e:[0,1]\rightarrow\sigma$ be an embedding of an analytic, oriented path into the spatial $D$-manifold $\sigma$. Then 
\begin{equation}
\eformfactor{e}{a}{x}=\INT{0}{1}\dx{t}\dot{e}^a(t)\diracdeltad{D}{x,e(t)}
\end{equation}
is a called distributional form factor for the edge $e$. The dot indicates the derivative with respect to the curve parameter $t$. 
This can be used to smear one-forms along this specific edge.
\item Let $S:[0,1]^{D-1}\rightarrow\sigma$ be an embedding of an analytic, oriented surface into $\sigma$. Then
\begin{equation}
\Sformfactor{S}{a}{x}=\INT{0}{1}\dx{t_1}\INT{0}{1}\dx{t_2}\dots\INT{0}{1}\dx{t_{D-1}} \epsilon_{a a_1 a_2 \dots a_{D-1}}S^{a_1}_{,t_1}(t) S^{a_2}_{,t_2}(t) \dots S_{,t_{D-1}}^{a_{D-1}}(t)\diracdeltad{D}{x,S(t)}
\end{equation} 
is called the distributional form factor for the surface $S$. Here, $t=(t_1,t_2,\dots,t_{D-1})$ is a  right-handed parametrization of $S$ and the comma indicates the derivative with respect to the parametrization, {\ie} $S^a_{,t_i}=\partialderivative{S^a}{t_i}$. This can be used to smear vector densities over this specific surface.
\end{enumerate}
\end{definition}

In terms of form factors for edges, the smeared connection is
\begin{equation}
A(e)={\int\!}_\sigma\dxd{D}{x} A_a(x)\eformfactor{e}{a}{x}.
\end{equation}
The same can be done with fluxes:
\begin{equation}
E(S)={\int\!}_\sigma\dxd{D}{x} E^a(x) \Sformfactor{S}{a}{x}.
\end{equation}

The oriented intersection number can be deduced from the commutator of holonomy and flux, expressed via form factors. This yields
\begin{equation}
\orientedintersection{e}{S}=\frac{1}{i}\bigcommutator{A(e)}{\fluxS}={\int\!}_\sigma\dxd{D}{x} \eformfactor{e}{a}{x}\Sformfactor{S}{a}{x}
\end{equation}

Having established the above notion of distributional form factors, we want to construct a vector space. In principle, each surface or edge is described by a $D$-tuple of individual form factors, to which we will refer to by omitting the index, {\ie} $F_e$ or $F_S$. 
Since we can add distributions and multiply them by scalars, it is trivial to see that they possess a vector space structure. In this sense we can look at finite linear combinations of form factors, such as
\begin{equation}
F:= n_1F_{e_1}+n_2F_{e_2}+\dots+n_kF_{e_k}, \qquad n_k\in\integernumbers
\end{equation} 
for some edges. The index structure of the form factors now tells us that it is not very reasonable to add form factors for edges and surfaces. We can only add form factors of the same kind if we want to integrate them against objects that have an opposite index structure. 

In principle, we would like to consider only linear combinations with respect to integer scalars as indicated above, since they could be interpreted as representation labels of $\Uone$. However, when doing so, we run into a problem. The integer numbers $\integernumbers$ are not a field since they lack a multiplicative inverse. As a consequence it is not possible to have a vector space over the integers. 
When we want to work at the level of integers, the structure we have at hand is a $\integernumbers$-module. In order to be able to work with such an object when considering Weyl algebras and quasifree states we would have to generalize Weyl algebras to symplectic $\integernumbers$-modules, which we do not want to do here.
Therefore, we generalize the linear combinations to coefficients in $\realnumbers$. 
The price to pay is that we actually cannot work with $\Uone$ anymore but rather have to deal with the Bohr compactification of the real line, $\realnumbers_\mathrm{Bohr}$. However, the integer linear combinations of form factors are part of the symplectic vector space corresponding to $\realnumbers_\mathrm{Bohr}$, and hence it is at this point convenient to work in this setup.
Furthermore, a possible Weyl algebra constructed subsequently will be a $\realnumbers_\mathrm{Bohr}$ HF Weyl algebra, and hence a $\Uone$ HF Weyl algebra will be embedded into this larger algebra. Every representation of the $\realnumbers_\mathrm{Bohr}$ algebra gives rise to a representation of the $\Uone$ algebra. We will consider a $\Uone$ HF algebra that is embedded into a $\realnumbers_\mathrm{Bohr}$ HF algebra and consider only representations of the $\Uone$ algebra that come from the $\realnumbers_\mathrm{Bohr}$ representation. 

We define different vector spaces of distributional form factors as follows.
\begin{definition}
\label{def_VSformfactors}
Let $F_e$ and $F_S$ be distributional form factors for edges and surfaces.
\begin{enumerate}
\item The vector space of distributional form factors of paths is denoted by
\begin{equation}
H_\text{edge}=\set{\SUM{k=1}{m}\lambda_k F_{e_k}}{m<\infty,~\lambda_k\in\realnumbers}.
\end{equation}
\item The vector space of distributional form factors of surfaces is denoted by
\begin{equation}
H_\text{surface}=\set{\SUM{k=1}{m}\lambda_k F_{S_k}}{m<\infty,~\lambda_k\in\realnumbers}.
\end{equation}
\item The vector space of distributional form factors of paths and surfaces is then
\begin{equation}
H=H_\text{edge}\oplus H_\text{surface}.
\end{equation}
The elements in $H$ are denoted by $(e,S)$ and stand for the pair of linear combinations of form factors $F_e\in H_\text{edge}$ and $F_S\in H_\text{surface}$, which correspond to collections of edges and surfaces, respectively. 
\end{enumerate}
\end{definition}

By the addition of  form factor tuples such as
\begin{equation}
(e_1,S_1)+(e_2,S_2)=(e_1+e_2,S_1+S_2)
\end{equation}
we actually mean the addition of the form factors for the individual distributions for the edges and surfaces. The neutral element of the addition, needed for the vector space structure, is inherently given by the form factor that is constantly zero everywhere, because there is no edge or surface, denoted by $(0,0)$. 
Scalar multiplication will be interpreted as multiplying the individual form factors by a scalar and is denoted as 
\begin{equation}
\lambda(e,S)=(\lambda e,\lambda S).
\end{equation} 
What we also need for a vector space is the inverse of addition. This can simply be set to 
\begin{equation}
(-e,-S)=-(e,S).
\end{equation}
The appearing minus sign has a slightly more convenient interpretation, which can be seen when looking for the form factor for the inverse of the generic edge $e$,
\begin{equation}
\eformfactor{e\INV}{a}{x}=-\eformfactor{e}{a}{x}
\end{equation}
and hence the additive inverse for edges only is given by the form factor for the inverse edge.

The same holds true for surfaces. Multiplying the form factor by a minus sign can be interpreted as a change of orientation. 

Given this vector space structure, we can look at the oriented intersection number that becomes a map
\begin{equation}
I:H_\text{edge}\oplus H_\text{surface}\rightarrow \realnumbers,\qquad (e,S)\mapsto \orientedintersection{e}{S}.
\end{equation}
As can easily be seen from the integral form of the intersection number, it is linear in both edges and surfaces in the following sense: 
\begin{equation}
\orientedintersection{e+\lambda e\pr}{S+\mu S\pr}
=\orientedintersection{e}{S}
+\mu\,\orientedintersection{e}{S\pr}
+\lambda\,\orientedintersection{e\pr}{S}
+\lambda\mu\,\orientedintersection{e\pr}{S\pr}\vphantom{\INTset{\sigma}\dxd{D}{x}},
\end{equation}
with $\lambda, \mu \in \realnumbers$. 

For concatenated edges of the form $e=e_2\circ e_1$, the form factors separate according to
\begin{equation}
\eformfactor{e}{a}{x}=\eformfactor{e_1}{a}{x}+\eformfactor{e_2}{a}{x}\vphantom{\INT{0}{1}}.
\end{equation}
Consequently, the smeared connections break up in parts belonging to the segments, {\ie} 
\begin{equation}
A(e)=A(e_1)+A(e_2).
\end{equation}
This reproduces the decomposition property for holonomies as 
\begin{equation}
\holonomy=e^{iA(e)}=e^{iA(e_2)+iA(e_1)}=e^{iA(e_2)}e^{iA(e_1)}=\holonomyopen{e_2}\holonomyopen{e_1}.
\end{equation}
For surfaces a similar behavior arises. A sum of form factors, such that the surfaces consist of a set of faces $\{S_k\}$, decomposes and for the electric fluxes we find
\begin{equation}
E(S)=\SUM{k}{} E(S_k).
\end{equation} 

\subsection{Construction of the Weyl algebra}
\label{subsec_Uone_ConstrWeylAlg}
We want to construct a Weyl algebra for the structure group $\Uone$. 
Some of the results we present here were first obtained for the three-dimensional case in \cite{Nekovar14}.

\begin{proposition}\label{prop_nondeg}
The presymplectic form 
\begin{equation}
\symp{\espd{}}{\espu{\prime}}=\oisec{e}{S\pr}-\oisec{e\pr}{S}
\end{equation}
is nondegenerate, {\ie}
\begin{equation}\label{eq_Uone_PropNonDegtwo}
\symp{\espd{}}{\espu{\prime}}=0\qquad\forall \espu{\prime}\in H 
\end{equation}
implies that $\espu{}=(0,0)$.
\end{proposition}

\begin{proof}
We want to argue by contradiction and thus suppose that Eq. \eqref{eq_Uone_PropNonDegtwo} holds true and $(e,S)\neq (0,0)$.
Since the presymplectic form is nondegenerate, it must hold that 
\begin{equation}\label{eq_NonDeg_Oisec}
\oisec{e}{S\pr}=\oisec{e\pr}{S}
\end{equation}
for all pairs $\espu{\prime}$.

We have to distinguish two cases. Let us suppose that $\oisec{e}{S}\neq 0$. The choice of the pair of edge and surface that leads into a contradiction is
\begin{equation}
\espu{\prime}=(-e,S)\quad\text{or}\quad \espu{\prime}=(e,-S)
\end{equation}
which is just the original element, but with {\eg} inverse orientation for the edge. 
Equation \eqref{eq_NonDeg_Oisec} hence turns into
\begin{equation}
\oisec{e}{S}=-\oisec{e}{S}.
\end{equation}
This is a contradiction, because we required the intersection number to be nonvanishing. 

Second, the situation in which the intersection number of $e$ and $S$ is vanishing has to be considered: $\oisec{e}{S}=0$. Including the case that $e=0$, we take a look at the   
situation 
\begin{equation}
\espu{\prime}=(e\pr,S),
\end{equation}
where we keep $e\pr$ generic but fix the surface part to $S\pr=S\neq 0$. This results in 
\begin{equation}
0=\oisec{e}{S}=\oisec{e\pr}{S}\qquad\forall e\pr.
\end{equation}
There is certainly an (infinitesimal) edge $e\pr$ that intersect the nontrivial surface $S$. This contradicts our assumption. If we want to include a vanishing surface $S=0$, we fix $e\pr=e\neq 0$ with generic $S\pr$:
\begin{equation}
\espu{\prime}=(e,S\pr).
\end{equation}
Then we chose an (infinitesimal) surface $S\pr$ intersected by $e$, s.\,t. 
\begin{equation}
0=\oisec{e}{S}=\oisec{e}{S\pr}\qquad\forall S\pr
\end{equation}
yields a contradiction. 

After all, every assumption that $\espd{}$ is nontrivial resulted in a contradiction, which implies that 
\begin{equation}
\espd{}=0.
\end{equation}
\end{proof}

We showed that the presymplectic form $\sigma$ is in fact a symplectic form. Therefore, the corresponding Weyl algebra will be uniquely defined; see section \ref{se:weyl}.
Now we can state a definition of the Weyl algebra corresponding to the canonical commutation relation \eqref{eq_WeylUone_CommutatorConnectionElectricfield}.

\begin{definition}
\label{def_UoneHolFluxWeylAlg}
The CCR algebra over the presymplectic space $(H,\sigma)$, with $H$ the vector space of pairs of distributional form factors for edges and surfaces and $\symp{(e_1,S_1)}{(e_2,S_2)}=\orientedintersection{e_1}{S_2}-\orientedintersection{e_2}{S_1}$, generated by Weyl elements $\weyl{e}{S}$ that obey the Weyl relations 
\begin{equation}
\begin{aligned}
\weyl{e_1}{S_1}\weyl{e_2}{S_2}&=e^{-\frac{i}{2}\symp{(e_1,S_1)}{(e_2,S_2)}}\weyl{e_1+e_2}{S_1+S_2},\\
\weyl{e}{S}\DAGGER&=\weyl{-e}{-S},
\end{aligned}
\end{equation}
is called $\Uone$ HF Weyl algebra. It is denoted as $\CCRW{H,\sigma}$.
\end{definition}
It should be pointed out that according to the discussion preceding definition \ref{def_VSformfactors} we consider only Weyl elements that belong to the $\integernumbers$-module contained in $H$ as physically relevant.

Now we want to prove our statement from the beginning of this section.
\begin{proof}[Proof of proposition \ref{prop_UoneWeylRepHF}]
Let $\varphi$ be a state on $\CCRW{H,\sigma}$ with GNS representation $\GNS$ and let it be regular in the surfaces. That is, the field operator
\begin{equation}
    \rep{E(S)}:=\frac{1}{i}\derivative{}{t}\rep{\weyl{0}{tS}}\big|_{t=0}
\end{equation}
exists for all $(0,S)\in H$. Furthermore, we set
\begin{equation}
    \rep{\holonomy}:=\rep{\weyl{e}{0}}.
\end{equation}
In this way, $\pi$ becomes a representation of the algebra $\afree$. We compute the respective commutators: 
\begin{equation}\label{eq_CommRepHolHol}
    \commutator{\rep{\holonomyopen{e}}}{\rep{\holonomyopen{e\pr}}}=\commutator{\rep{\weyl{e}{0}}}{\rep{\weyl{e\pr}{0}}}=0.
\end{equation}
For two surface field operators we find 
\begin{equation}\label{eq_CommRepFlFl}
    \commutator{\rep{E(S)}}{\rep{E(S\pr)}}=\frac{1}{i^2}\derivative{^2}{t \mathrm{d}t\pr}\commutator{\rep{\weyl{0}{tS}}}{\rep{\weyl{0}{t\pr S\pr}}}\big|_{t,t\pr=0}=0.
\end{equation}
The last commutator is 
\begin{equation}\label{eq_CommRepHolFl}
\begin{aligned}
    \commutator{\rep{\holonomy}}{\rep{E(S)}}&=\frac{1}{i}\derivative{}{t}\commutator{\rep{\weyl{e}{0}}}{\rep{\weyl{0}{tS}}}\big|_{t=0}\\
    &=\frac{1}{i}\derivative{}{t}\roundbrackets{e^{-\frac{i}{2}\oisec{e}{S}t}-e^{\frac{i}{2}\oisec{e}{S}t}}\rep{\weyl{e}{tS}}\big|_{t=0}\\
    &=-\oisec{e}{S}\rep{\holonomy}.
\end{aligned}    
\end{equation}
Equations \eqref{eq_CommRepHolHol} and \eqref{eq_CommRepHolFl}, together with the linearity of $\pi$  establish that the ideals \eqref{eq:g1}--\eqref{eq:g4} are in the kernel of $\pi$, and thus $\pi$ is a representation of the HF algebra $\hfalgebra$ for $G=\Uone$. 
\end{proof}

This result suggests to identify Weyl elements with vanishing surface contribution with holonomies and the field operators corresponding to the surfaces with fluxes. We can decompose a general Weyl element as follows:
\begin{equation}
    \rep{\weyl{e}{S}}=e^{\frac{i}{2}\oisec{e}{S}}\rep{\weyl{e}{0}}\rep{\weyl{0}{S}}=e^{\frac{i}{2}\oisec{e}{S}}\rep{\holonomy}e^{i\pi(E(S))}.
\end{equation}
The Weyl elements combine holonomies with exponentiated fluxes. 

For full loop quantum gravity, we have to consider the structure group $\SUtwo$. The main difference to the $\Uone$ case is the noncommutativity of fluxes and that holonomies are elements of $\SUtwo$. For holonomies it is only possible to combine two of them in a single exponential if they represent parts of concatenated edges. Exponentiated fluxes can in general not be combined in a single exponential describing a sum of fluxes, since the commutator of fluxes does not yield a flux again. 

Another problem is the localization of the commutation relation. In the $\Uone$ case all information about the intersection behavior of edge and surface, coming from the commutation relation, is encoded in the exponential of the oriented intersection number and can be separated from the actual holonomy and flux. This is not possible for $\SUtwo$. Here, the commutator inserts into the holonomy $\sutwo$-valued functions at the intersection point with the surface. The exponentiated intersection number becomes a complicated combination of segments of the holonomy and $\SUtwo$ elements at the intersection points. 

As a result, it is impossible to extract the symplectic form from this kind of expressions. It is equally not manageable to directly construct a general symplectic form from the commutation relation. Therefore, it seems implausible that a Weyl algebra formulation of LQG exists. These problems notwithstanding, in section \ref{se:sutwo} we will find a \emph{representation} of the HF algebra of LQG which resembles a representation of an almost quasifree state, using the general result from Sec. \ref{sec_generalresult}.

\section{Almost quasifree states and representations}\label{se:aqfree}
After having established this Weyl algebra, we want to introduce a new state on it that gives rise to representations of the holonomy-flux algebra in accordance with Sec. \ref{sec_generalresult}.

What distinguishes the CCR of $\Uone$ LQG from other field theories is its purely topological nature. Holonomies and fluxes do not commute if there is an intersection point between the corresponding edge and surface. 
The inequality \eqref{eq_inequality} turns into
\begin{equation}\label{eq_inequality_eS}
\sqrt{\cov{\espu{}}{\espu{}}}\sqrt{\cov{\espu{\prime}}{\espu{\prime}}}\geq \frac{1}{2}\lvert{\symp{\espu{}}{\espu{\prime}}}\rvert\qquad \forall \espu{},\espu{\prime} \in H.
\end{equation}
It seems unlikely that it is possible to find a quasifree state for the $\Uone$ HF Weyl algebra, and LQG in general, for two reasons: 
First of all, the right-hand side of \eqref{eq_inequality_eS} is diffeomorphism invariant while the left-hand side cannot be.\footnote{Indeed, assume a diffeomorphism invariant $\alpha$ and consider $S=0,e'=0$. Then 
\begin{equation*}
\frac{1}{2}\lvert{\oisec{\varphi(e)}{S'}}\rvert=\frac{1}{2}\lvert{\symp{(\varphi(e),0)}{(0,S')}}\rvert \leq \sqrt{\cov{(\varphi(e),0)}{(\varphi(e),0)}}\sqrt{\cov{(0,S')}{(0,S')}}
=\sqrt{\cov{(e,0)}{(e,0)}}\sqrt{\cov{(0,S')}{(0,S')}}
\end{equation*}
where $\varphi$ is any diffeomorphism. One can construct $\varphi$ that makes the left-hand side arbitrarily large, leading to a contradiction.} This means that \eqref{eq_inequality_eS} implies many further inequalities and is hence more restrictive than it might appear. Second, the right-hand side contains intersection numbers between $S$ and $e'$, and vice versa, between $S'$ and $e$. In contrast, the left-hand side depends on $\alpha$ evaluated on $(e,S)$ alone and on $(e',S')$ alone, respectively. Thus one factor on the left-hand side does not know about the argument of the other, yet together they have to bound a quantity depending on the relative position of the two arguments.  

In the framework of projective loop quantum gravity \cite{Lanery:2014bca,Lanery:2015mxa}, there is a result about exactly the type of inequalities needed for quasifree states. It states that there cannot be a nonsingular covariance that satisfies an inequality similar to \eqref{eq_inequality}, and hence there cannot exist such states. This seems to be a feature that comes from the structure of the underlying algebra itself. 

Nevertheless, it is possible to construct  a new type of state for the $\Uone$ HF Weyl algebra, which can be interpreted as a hybrid of the Ashtekar-Lewandowki state for the HF algebra and a quasifree state for Weyl algebras.  

\subsection{Almost quasifree states for the $\Uone$ holonomy-flux Weyl algebra}

In addition to the Gaussian fluctuations of quasifree states, we want to include a controllable peak position into the new state. This leads, alongside the fluctuations given by two-point correlation functions, to nonvanishing one-point correlation functions, {\ie} a condensate contribution.

\begin{proposition}
\label{def_AlmostQfreeStates}
For a symmetric, positive semidefinite bilinear form (covariance) $\mathbf{\alpha}:H_\text{surface}\times H_\text{surface} \rightarrow \realnumbers$ and a linear function (condensate contribution) $\mathbf{\beta}:H_\text{surface} \rightarrow\realnumbers$ with $\mathbf{\beta}(0)=0$, the linear functional 
\begin{equation}\label{eq_UoneAlmostQfreeSurf}
\aqfsstate{\weyl{e}{S}}=\delta_{e,0}e^{-\frac{1}{2}\covs{S}{S}}e^{i\cons{S}}
\end{equation}
is a state on $\CCRW{H,\sigma}$, which is called almost quasifree.
\end{proposition}

\begin{proof}
As a first step, note that the Weyl elements form a basis of $\CCRW{H,\sigma}$; see Lemma \ref{le:linear_independence}. Thus 
\eqref{eq_UoneAlmostQfreeSurf} defines $\aqfsstate{\cdot}$ as a linear form on all of $\CCRW{H,\sigma}$. 

In order to show that $\aqfss$ is a state, we show that it is normalized and positive. We have 
\begin{equation}
\aqfsstate{W(0,0)}=\delta_{0,0}e^{-\frac{1}{2}\covs{0}{0}}e^{i\mathbf{\beta}(0)}=1.
\end{equation}
For algebra elements $x=\SUM{k=1}{n}b_k\weyl{e_k}{S_k}$, with $n\in\mathbb{N}$ and $b_k\in\complexnumbers$, we have that 
\begin{equation}
\begin{aligned}
    \aqfss(xx\AST)&=\SUM{j,k=1}{n}b_j\bar{b_k}\aqfsstate{\weyl{e_j}{S_j}\weyl{e_k}{S_k}}\\
    &=\SUM{j,k=1}{n}b_j\bar{b_k}\aqfsstate{\weyl{e_j-e_k}{S_j-S_k}}e^{-\frac{i }{2}\symp{\espd{j}}{\espd{k}}}\\
    &=\SUM{j,k=1}{n}b_j\bar{b_k}\delta_{e_j,e_k}e^{-\frac{1}{2}\covs{S_j-S_k}{S_j-S_k}}e^{i\mathbf{\beta}(S_j-S_k)}e^{-\frac{i}{2}\symp{\espd{j}}{\espd{k}}}\\
    &=\SUM{j,k=1}{n}c_j\bar{c_k}\delta_{e_j,e_k}e^{-\frac{1}{2}\covs{S_j-S_k}{S_j-S_k}}e^{-\frac{i}{2}\symp{\espd{j}}{\espd{k}}}
\end{aligned}    
\end{equation}
and introduced in the last line the coefficients $c_j=b_je^{i\mathbb{\beta}(S)}$. The resulting sum can be split into three different contributions, depending on the configuration of edges. 

The sum can run over pairs of nontrivial edges, so all $e_j\neq0$. 
With this 
\begin{equation}
\begin{aligned}
&\SUM{j,k=1}{n}c_j\bar{c_k}\delta_{e_j,e_k}e^{-\frac{1}{2}\covs{S_j-S_k}{S_j-S_k}}e^{-\frac{i}{2}\symp{\espd{j}}{\espd{k}}}=\\
&=\SUM{j=1}{n}c_j\bar{c_j}e^{-\frac{1}{2}\covs{S_j-S_j}{S_j-S_j}}e^{-\frac{i}{2}\symp{\espd{j}}{\espd{j}}}=\SUM{j}{}\absvalue{c_j}^2\geq0,
\end{aligned}    
\end{equation}
where we used the Kronecker delta to collapse the sum over $k$, the symplectic form is antisymmetric and hence its diagonal elements vanish. 

The second contribution are terms where nontrivial and trivial edges meet in the Kronecker delta. These terms vanish immediately. 

For the third contribution, we have to consider terms with only trivial edges. Hence, the Kronecker delta gives one and the sums do not collapse. Realizing that the symplectic form does also vanish, this yields 
\begin{equation}
\SUM{j,k=1}{n}c_j\bar{c}_k 
e^{-\frac{1}{2}\bigroundbrackets{\covs{S_j}{S_j}+\covs{S_k}{S_k}-2\covs{S_j}{S_k}}}=\SUM{j,k=1}{n}d_j\bar{d}_k 
e^{\covs{S_j}{S_k}}.
\end{equation}
Here, we introduced $d_j=c_je^{-\frac{1}{2}\covs{S_j}{S_j}}$. The Hadamard product, as well as the sum, of two positive semidefinite $n\times n$ matrices is also positive semidefinite. Therefore the matrix obtained by exponentiating the components $\covs{S_j}{S_k}$ of the covariance is positive semidefinite as well. Hence,
\begin{equation}
    \SUM{j,k=1}{n}d_j\bar{d}_k 
e^{\covs{S_j}{S_k}}\geq0
\end{equation}
and the positivity of the state, {\ie} $\aqfsstate{xx\AST}\geq0$, follows.
\end{proof}

The almost quasifree state $\aqfss$ behaves as the AL state for holonomies, {\ie} is only nonvanishing if there are none, but shows a different behavior concerning the fluxes. The flux part is a Gaussian function. Being differentiable only in the surface variable, there are only field operators corresponding to fluxes. Similar to the AL representation, there cannot be field operators for the connection. 

Given the GNS representation $\GNS$ of the almost quasifree state $\varphi$, the representation of fluxes is determined by Stone's theorem:
\begin{equation}\label{eq_UoneFluxClassRep}
\rep{E(S)}=\frac{1}{i}\derivative{}{t}\rep{\weyl{0}{tS}}\Big|_{t=0}.
\end{equation}
The holonomies can directly be represented by the Weyl elements:
\begin{equation}
\rep{\holonomy}=\rep{\weyl{e}{0}}.
\end{equation}
Similar to the situation for quasifree states, the commutator of holonomies and fluxes is given by the symplectic form. Hence, we find that (cf.~\eqref{eq_commutator_Weyl_field})
\begin{equation}
\commutator{\rep{\holonomy}}{\rep{E(S)}}=-\oisec{e}{S}\rep{\holonomy},
\end{equation}
such that this is in fact a representation of $\CCRW{H,\sigma}$. 

While the vacuum expectation values containing holonomies still vanish, the one- and two-point functions of fluxes do not anymore. 
Using the representation \eqref{eq_UoneFluxClassRep} and the state \eqref{eq_UoneAlmostQfreeSurf}, the peak for fluxes lies at 
\begin{equation}
\vacexp{\rep{E(S)}}=\frac{1}{i}\derivative{}{t}\state{\weyl{0}{tS}}\Big|_{t=0}=\cons{S},  
\end{equation}
{\ie} the condensate contribution $\conssym$. The two-point functions, describing the fluctuations include the surface covariance,
\begin{equation}
    \vacexp{\rep{E(S_1)}\rep{E(S_2)}}=\frac{1}{i^2}\derivative{^2}{t_1\dx{t_2}}\statewrt{S}{\weyl{0}{t_1S_1+t_2S_2}}\Big|_{t_1,t_2=0}=\covs{S_1}{S_2}+\cons{S_1}\cons{S_2}.
\end{equation}
Similarly, all higher $n$-point functions decompose into products of covariances and condensate contributions. For vanishing $\conssym$, only the two-point functions contribute and reproduce the behavior of quasifree states.

Because of the GNS construction, each choice of $\covssym$ and $\conssym$ yields, up to isomorphisms, a different almost quasifree representation of the holonomy-flux Weyl algebra. 
Since there is no default object like the oriented intersection number $H_\text{surface}$ that allows for the definition of an inner product,
one has to use additional structures.

\subsection{Almost quasifree representation with respect to a Fock space}\label{se:aqfree_fock}
\label{subsec_Uone_FockRepAqf}
In this and the following section we want to introduce almost quasifree representations of $\CCRW{H,\sigma}$ and the the $\Uone$ HF algebra. Similar to the nature of the almost quasifree state, which resembles in some aspects the AL state, the representations will be constructed as an augmented version of the AL representation. 

We look at a bosonic Fock space
\begin{equation}
\curlyF=\bigotimes\limits_{k=0}^\infty \curlyh^{\otimes k}.
\end{equation}
Here, $\curlyh$ is any Hilbert space. It plays the role of one-particle Hilbert space.
For $f\in\curlyh$, the creation and annihilation operators of this field theory are subject to the commutation relations 
\begin{equation}
\begin{aligned}
&\commutator{\anniF{f}}{\creaF{g}}=\innerproductof{\curlyh}{f}{g}\unitelementof{\curlyF},\\
&\commutator{\anniF{f}}{\anniF{g}\vphantom{\creaF{f}}}=0=\commutator{\creaF{f}}{\creaF{g}}.
\end{aligned}
\end{equation}
With this, the self-adjoint field operators of the theory are denoted by
\begin{equation}
\label{eq_scalarfield}
\sfield{f}=\creaF{f}+\anniF{f}.
\end{equation}
The commutation relations for creation and annihilation operators extend to $\sfield{f}$, {\ie} 
\begin{equation}
\commutator{\sfield{f}}{\sfield{g}}=2i\mathrm{Im}\roundbrackets{\innerproductof{\curlyh}{f}{g}}. 
\end{equation} 
For general one-particle Hilbert space elements $f$ and $g$, this is not forced to  vanish. However, for well-defined field operators that are smeared with real test functions the commutator vanishes identically.  
The two-point functions with respect to the Fock vacuum $\vac{0}$ are 
\begin{equation}
\vacexpof{0}{\sfield{f}\sfield{g}}=\innerproductof{\curlyh}{f}{g}.
\end{equation}

The total Hilbert space we want to consider for almost quasifree states is now the tensor product space 
\begin{equation}
\HSFock=\ALHS\otimes\curlyF
\end{equation}
with the cyclic vacuum state 
\begin{equation}
\vac{\curlyF}=\ALvac\otimes\vac{0}.
\end{equation}

Instead of giving a representation of the Weyl elements right from the beginning, we present the representations of holonomies and fluxes, and show that this gives rise to the desired representation. 
Let us suppose we have a real and linear map
\begin{equation}
\Gamma: H_\text{surface}\rightarrow\curlyh,
\end{equation}
which allows one to relate a surface and hence its form factor to an element of the one-particle Hilbert space. The linearity is necessary in order to have $\Gamma(S+S\pr)=\Gamma(S)+\Gamma(S\pr)$, such that the field operators are linear in the surfaces and behave similar to the AL flux operators. We want to consider the following representation of holonomies and fluxes on the tensor product Hilbert space:
\begin{equation}
\begin{aligned}
\repF{\holonomy}&=\holonomy\otimes \unitelementof{\curlyF},\\
\repF{E(S)}&=X_S\otimes\unitelementof{\curlyF}+\unitelementof{\AL}\otimes\left(\sfield{\Gamma(S)}+\cfluxS\right).
\end{aligned}
\end{equation}
The representation is characterized by the flux through a classical background electric field $E^{(0)}$
\begin{equation}
\cfluxS=\INTset{\sigma}\dxd{D}{x}E^{(0)a}(x)\Sformfactor{S}{a}{x},
\end{equation}
which gives the peak position of the Gaussian part of the almost quasifree state.

In order to check if this is a representation of the HF algebra we consider the commutation relations of the object defined above. For the fluxes we find
\begin{equation}
\commutator{\repF{E(S)}}{\repF{E(S\pr)}}
=\commutator{X_{S}}{X_{S\pr}}\otimes\unitelementof{\curlyF}+\unitelementof{\AL}\otimes\commutator{\sfield{\Gamma(S)}}{\sfield{\Gamma(S\pr)}}
=0,
\end{equation}    
since the flux operators still commute and for real $\Gamma$ the commutator of field operators vanishes, too. Also the commutator of holonomies still vanishes:
\begin{equation}
\commutator{\repF{\holonomy}}{\repF{\holonomyopen{e\pr}}}=\commutator{\holonomy}{\holonomyopen{e\pr}\vphantom{\repF{\holonomy}}}\otimes\unitelementof{\curlyF}=0.
\end{equation} 
The commutator of holonomy and flux representations is also recovered as
\begin{equation}
\commutator{\repF{\holonomy}}{\repF{E(S)}}=\commutator{\vphantom{\repF{\holonomy}}\holonomy}{X_S}\otimes\unitelementof{\curlyF}=\repF{\commutator{\holonomy}{E(S)}}.
\end{equation} 
We indeed found a representation of the HF algebra on $\HSFock=\ALHS\otimes\curlyF$. 

We find that 
\begin{equation}
\FflexpS:=e^{i\repF{E(S)}}=e^{iX_S}\otimes e^{i\sfield{\Gamma(S)}+i\cfluxS}
\end{equation}
and set
\begin{equation}
\Fhol=\holonomy\otimes\unitelementof{\curlyF}.
\end{equation} 
With this, we can determine commutation relations for $\Fhol$ and $\FflexpS$ and hence define Weyl elements
\begin{equation}\label{eq_Uone_FWeylOperator}
\weylF{e}{S}=e^{\frac{i}{2}\orientedintersection{e}{S}}\Fhol\FflexpS=e^{\frac{i}{2}\orientedintersection{e}{S}} \holonomy e^{iX_S}\otimes e^{i\sfield{\Gamma(S}+i\cfluxS},
\end{equation}
which satisfy the Weyl relations 
\begin{equation}
\begin{aligned}
\weylF{e_1}{S_1}\weylF{e_2}{S_2}
&= 
e^{-\frac{i}{2}\bigroundbrackets{\orientedintersection{e_1}{S_2}-\orientedintersection{e_2}{S_1}}}  \weylF{e_1+e_2}{S_1+S_2} 
,\\ \weylF{e}{S}\DAGGER&=\weylF{-e}{-S}.
\end{aligned}
\end{equation}
 
When considering the corresponding vacuum expectation value, the inner product splits into two contributions and hence the vacuum expectation value splits into
\begin{equation}
\vacexpof{\curlyF}{\weylF{e}{S}}=\vacexpof{\AL}{e^{i\oisec{e}{S}}\holonomy e^{iX_S}}\vacexpof{0}{e^{i\sfield{\Gamma(S)}+i\cfluxS}}.
\end{equation}  
The vacuum expectation value on $\ALHS$ is given by 
\begin{equation}
\vacexpof{\AL}{e^{i\oisec{e}{S}}\holonomy e^{iX_S}}=e^{i\oisec{e}{S}}\delta_{e,0}=\delta_{e,0}.
\end{equation} 

Using the Baker-Campbell-Hausdorff decomposition of the annihilation and creation parts of $\sfield{\Gamma(S)}$, we find 
\begin{equation}
\vacexpof{0}{e^{i\sfield{\Gamma(S)}+i\cfluxS}}=e^{-\frac{1}{2}\innerproductof{\curlyh}{\Gamma(S)}{\Gamma(S)}}e^{i\cfluxS}.
\end{equation} 
Consequently, when combining the contributions, we again end up with an expression that is singular in the fluxes but Gaussian in the surfaces:
\begin{equation}
\vacexpof{\curlyF}{\weylF{e}{S}}=\delta_{e,0}e^{-\frac{1}{2}\innerproductof{\curlyh}{\Gamma(S)}{\Gamma(S)}}e^{i\cfluxS}.
\end{equation} 
 
This motivates the interpretation of the just presented representation of the HF algebra as an almost quasifree representation of the HF Weyl algebra determined by the covariance 
\begin{equation}
\covs{S}{S\pr}=\innerproductof{\curlyh}{\Gamma(S)}{\Gamma(S\pr)}
\end{equation} 
and the condensate contribution $\cfluxS$.

\subsection{A specific example of almost quasifree states}
\label{subsec_Uone_FirstRepAqf}
We want to give a more detailed example of the representation that was introduced in the previous section. 
The Hilbert space for this representation is 
\begin{equation}
\HScf=\ALHS\otimes\ltwo,
\end{equation}
{\ie} the AL Hilbert space times the Hilbert space of square summable, complex sequences, which is the Hilbert space of the harmonic oscillator. 
The orthonormal basis of $\ltwo$ is denoted by $\curlybrackets{\vac{n}}$, such that $\innerproduct{\vac{i}}{\vac{j}}=\delta_{ij}$.
The fluctuations and hence the covariance of the representation are determined by functions $f^a(x)$ integrated against a surface form factor, giving rise to the fluxlike object
\begin{equation}\label{eq_fluctuationfunction}
\flucS=\INTset{\sigma}\dxd{D}{x}f^a(x)\Sformfactor{S}{a}{x}.
\end{equation}

As the almost quasifree state behaves like the AL state for holonomies, we define the representation of holonomies accordingly:
\begin{equation}
\repcf{\holonomy}=\holonomy\otimes\unitelementof{\ltwo}.
\end{equation} 
Holonomies act only on the AL part of the Hilbert space. 
The flux representation on this Hilbert space is given by
\begin{equation}\label{eq_Uone_HOflux}
\repcf{E(S)}=\fluxvfS\otimes\unitelementof{\ltwo}+\unitelementof{\AL}\otimes\left(\flucS(\anni+\crea)+\cfluxS\right),
\end{equation}
with $a\equiv a(1), a^\dagger\equiv a^\dagger(1)$ the usual annihilation and creation operators of the harmonic oscillator, {\ie} $\commutator{\anni}{\crea}=1$.

The vacuum of $\HScf$ is the product of the vacua of the individual Hilbert spaces, which is
\begin{equation}
\vac{\curlyF}= \vac{\AL}\otimes\vac{0},
\end{equation} 
where $\vac{0}$ is the cyclic vector of the harmonic oscillator Hilbert space. 

In order to determine the vacuum expectation value of the Weyl operators we have to determine their action on $\vac{\curlyE}$. This is furthermore very instructive to see, since it tells us how the states generated by $\weylcf{e}{S}$ actually look like. Because of the tensor product structure we can look at the individual factors separately. We begin with the AL part. Every cylindrical function can be expressed in terms of holonomies, so it is sufficient to look only at a single holonomy that acts on the AL vacuum in order to understand the structure of Weyl operator generated states. It holds that 
\begin{equation}\label{eq_WeylOperatorALpart}
e^{\frac{i}{2}\oisec{e}{S}}\holonomy e^{iX_S}\vac{\AL}
=e^{\frac{i}{2}\oisec{e}{S}}\holonomy.
\end{equation}
This is the case because the AL vacuum state is basically the constant function $\vac{\AL}=1$ and hence is killed by acting on it with the derivative operator $X_S$. 
There is an additional phase factor, which is a remnant of the flux and knows about the intersection structure of $e$ and $S$.  
Therefore the AL-part of the Weyl operator creates cylindrical functions from the vacuum that have additional information about intersections with surfaces. 

For the harmonic oscillator part we make use of the notion of coherent states.
We realize that 
\begin{equation}
e^{i\flucS(\crea+\anni)}=e^{i\flucS\crea-\overline{i\flucS}\anni},
\end{equation}
which is the coherent state operator for the harmonic oscillator for a purely imaginary coherent state parameter $i\flucS$. 
Hence the action of this on the harmonic oscillator vacuum generates a coherent state with respect to the function $\flucS$, describing the fluctuations, {\ie}
\begin{equation}\label{eq_Uone_CScflux}
e^{i\flucS(\crea+\anni)}\vac{0}=e^{-\frac{1}{2}\flucS^2}\SUM{n=0}{\infty}\frac{(i\flucS)^n}{\sqrt{n!}}\bigroundbrackets{\crea}^n\vac{n}=:\vac{i\flucS}.
\end{equation} 

Putting both tensor factors together and including the condensate contribution yields the action of the Weyl operators on the vacuum:
\begin{equation}
\weylcf{e}{S}\vac{\curlyE}=e^{\frac{i}{2}\oisec{e}{S}}e^{i\cfluxS} \left(\holonomy\otimes\vac{i\flucS}\right).
\end{equation}
This is in fact a product of a holonomy encoding the edge $e$ and coherent state of the harmonic oscillator that knows about the classical flux through the surface $S$ and the corresponding fluctuations. Additionally the phase factor at the beginning is aware of the intersection structure of $e$ and $S$. 

Finally we can consider the vacuum expectation value. The inner product on the Hilbert space $\HScf=\ALHS\otimes \ltwo$ is given by the product of the ones of the individual Hilbert spaces, since there is no entanglement:
\begin{equation}
\innerproductof{\HScf}{\slotdot}{\slotdot}=\innerproductof{\ALHS}{\slotdot}{\slotdot}\innerproductof{\ltwo}{\slotdot}{\slotdot}.
\end{equation} 
Again, we look at the contributions separately. For holonomies and exponentiated fluxes we have 
\begin{equation}
\innerproductof{\ALHS}{\ALvac}{e^{\frac{i}{2}\oisec{e}{S}}\holonomy e^{iX_S}\ALvac}=\delta_{e,0}.
\end{equation}
Although there is a remnant of the surface in the vacuum expectation value, it does not contribute to it, since the expectation value is only nontrivial if the edge is trivial.  
The vacuum expectation value of harmonic oscillator coherent state operator is 
\begin{equation}
\innerproductof{\ltwo}{\vac{0}}{e^{i\flucS(\crea+\anni)}\vac{0}}
=e^{-\frac{1}{2}\flucS^2}.
\end{equation}
Putting things together yields the desired result. We find
\begin{equation}
\innerproductof{\HScf}{\vac{\curlyF}}{\weylcf{e}{S}\vac{\curlyF}}=\delta_{e,0}e^{-\frac{1}{2}\flucS^2}e^{i\cfluxS}.
\end{equation}
The vacuum expectation value of the Weyl operators hence is highly peaked at trivial holonomies and is a Gaussian for fluxes, peaked on the classical flux through the surface. 

Since the fluctuation function $\flucS$ is  linear in the form factors (see \eqref{eq_fluctuationfunction}),  we can set 
\begin{equation}\label{eq_Uone_cfluxCov}
\covs{S}{S\pr}=\flucS\flucSpr,
\end{equation}
which is bilinear and symmetric. 
Further, we can identify the classical flux $\cfluxS$ with a condensate contribution $\cons{S}$.

The almost quasifree state, corresponding to this covariance and condensate, is  given by
\begin{equation}\label{eq_Uone_AQfreeStatecflux}
\aqfsstate{\weyl{e}{S}}=\delta_{e,0}e^{-\frac{1}{2}\covs{S}{S}}e^{i\cons{S}}=\delta_{e,0}e^{-\frac{1}{2}\flucS^2}e^{i\cfluxS},
\end{equation} 
which matches the vacuum expectation value of the Weyl operators in the above representation. 
Therefore, by means of the GNS construction, $\GNSopen{\curlyF}$ is unitarily equivalent to the GNS representation of the almost quasifree state \eqref{eq_Uone_AQfreeStatecflux}. It can be interpreted as a representation of the almost quasifree state \eqref{eq_Uone_AQfreeStatecflux}. 

There is in fact a certain similarity to the Koslowski-Sahlmann (KS) representation and actually the considerations in \cite{Koslowski:2007kh,Sahlmann:2010hn}. There, the extension of the AL representation by a classical flux is considered. However, everything takes place on $\ALHS$ only. This allows for a shift of the flux peak, not for fluctuations in the fluxes. 

The representation we present here is an interesting example for the results about the representation theory of the HF algebra developed in \cite{Sahlmann:2002xu}. 
There, the main result is that a representation $(\curlyH,\pi)$ can be split up into a direct sum of representations
\begin{equation}
\begin{aligned}
\curlyH&\cong\bigoplus\limits_{\nu}\curlyH_\nu,\\
\pi&\cong\bigotimes\limits_\nu \pi_\nu,
\end{aligned}
\end{equation}
where the individual Hilbert spaces are $\curlyH_\nu\cong L_2(\overline{\curlyA},\grad{\mu_\nu})$. Denote the inclusion map $I_\nu:\curlyH_\nu\hookrightarrow\curlyH$. Under some rather mild assumptions it is possible to extend the representation such that --~adopted for our $\Uone$ considerations~-- one finds
\begin{equation}
\label{eq:generalrep}
\rep{E(S)}I_\nu(F)=I_\nu(X_S F)+\SUM{\iota}{}I_\iota(F_{\iota\nu}(S)).
\end{equation}
Here $F$ is a cylindrical function whose representation is $\repAL{F}=F$ and $F_{\iota\nu}(S)$ are functions that inherit certain properties from classical fluxes. It was always a bit mysterious what the physical meaning of a general representation in this class is, as the examples studied in the literature mostly have $\curlyH=L_2(\overline{\curlyA},\grad{\mu})$, i.e., no nontrivial direct sum. However, the example described above has this more complicated structure. Indeed, 
\begin{equation}
   \HScf=\ALHS\otimes\ltwo\simeq\bigoplus_{\mathbb{N}_0}\mathcal{H}_{\AL},
\end{equation}
and under this identification $\repcf{E(S)}$ takes the form \eqref{eq:generalrep} with
\begin{equation}
    F_{mn}= \cfluxS\, \delta_{m,n}+ \flucS\left(\delta_{m,n+1}+\delta_{m,n-1}\right), \qquad m,n\in \mathbb{N}_0. 
\end{equation}
This shows that such a nontrivial direct sum structure of a representation can have interesting physical significance. It seems to encode the fluctuations of the flux variables. 

\section{Almost quasifree representations for non-Abelian structure groups}\label{se:sutwo}
In this section we will consider representations of the HF algebra for non-Abelian $G$ that mimic the structure of the almost quasifree representations for $G=\Uone$. In particular, we will consider an example for the case most relevant for LQG, that of $G=\SUtwo$. It is straightforward to extend this construction to other structure groups but we will refrain from treating the general case. Rather, we will explore the changes to the area operator that are induced by the new representation. 

\subsection{An example for the structure group SU(2)}\label{se:sutwoqfree}
To obtain a representation of the HF algebra which mimicks the properties of an almost quasifree state, we will use proposition \ref{prop_flux_shift}, applied to the Hilbert space $\mathcal{K}=\curlyF$, where $\curlyF\equiv\curlyF(\curlyh)$ is the bosonic Fock space of a scalar field exactly as in Sec. \ref{se:aqfree_fock}. The representation to be modified is the 
AL representation; thus the resulting representation $\pi_\curlyF$ will act on the Hilbert space 
\begin{equation}
\HSFock=\ALHS\otimes\curlyF.
\end{equation} 
with the cyclic vector
\begin{equation}
\vac{\curlyF}=\ALvac\otimes\vac{0}.
\end{equation}
The AL vacuum is still the constant identity function and $\vac{0}$ is the Fock vacuum.

To construct the map $\Xi$ from Proposition \ref{prop_flux_shift} we define a classical background flux
\begin{equation}
    \cfluxSf=\INTset{S}\frac{1}{2}f^j(x){E^{(0)}}_j^a(x)\epsilon_{abc}\gradx{b}\wedge\gradx{c}
\end{equation}
which serves as the condensate contribution as in the $\Uone$ case. We define 
\begin{equation}\label{eq_SUtwofluxrep}
    \Xi(\fluxSf):= \sfield{\Gamma(S,f)}+\cfluxSf, 
\end{equation}
where $\phi$ is a scalar field as in \eqref{eq_scalarfield} and we require the map $\pair{S}{f}\mapsto \Gamma\pair{S}{f}$ again to be real and now linear in both, $S$ and $f$.
The representation of elementary algebra elements thus reads
\begin{equation}\label{eq_SUtwo_FockRep}
\begin{aligned}
\repF{\fluxSf}&=\fluxvfSf\otimes\unitelementof{\curlyF}+\unitelementof{\AL}\otimes\left(\sfield{\Gamma(S,f)}+\cfluxSf\right),\\
\repF{F}&=F\otimes\unitelementof{\curlyF}, 
\end{aligned}
\end{equation}
where $F$ is a cylindrical function. 

One of the characteristic features of the almost quasifree representations was factorization of flux operator $n$-point functions. This can also be directly recovered here. 
The one-point correlation function of flux operators is given by
\begin{equation}
\vacexpof{\curlyF}{\repF{\fluxSf}}=\cfluxSf,
\end{equation}
since the field one-point functions vanish.
Considering a product of two flux representations, {\ie}  
\begin{equation}
\begin{aligned}
&\repF{\fluxSf}\repF{\fluxSfopen{S\pr}{f\pr}}=\\
&=\roundbrackets{\FfluxSf}\roundbrackets{\FfluxSfpr}\\
&=\fluxvfSf\fluxvfSfpr\otimes\unitelementof{\curlyF}
+\fluxvfSf\otimes\left(\sfield{\Gamma\pair{S\pr}{f\pr}}+\cfluxSfpr\right)
+\fluxvfSfpr\otimes\left(\sfield{\Gamma\pair{S}{f}}+\cfluxSf\right)\\
&\quad+\unitelementof{\AL}\otimes \left(\sfield{\Gamma\pair{S}{f}}+\cfluxSf\right)\left(\sfield{\Gamma\pair{S\pr}{f\pr}}+\cfluxSfpr\right),
\end{aligned}
\end{equation}
we realize that, upon taking the vacuum expectation value, only the last term survives, since there is no operator that annihilates the AL vacuum. This yields the two-point correlation function
\begin{equation}
\vacexpof{\curlyF}{\repF{\fluxSf}\repF{\fluxSfpr}}=\innerproductof{\curlyh}{\Gamma\pair{S}{f}}{\Gamma\pair{S\pr}{f\pr}}+\cfluxSf\cfluxSfpr.
\end{equation}
 
For cylindrical functions only, there is no significant change to the pure AL representation, exactly as in the $\Uone$ case. 

Finally, to demonstrate that there are possible maps $\Gamma$, we will be specific about the scalar field that we added to the formalism. We consider the one-particle Hilbert space  
\begin{equation}
\curlyh=L^2(\realnumbers^3,\grad{^3x}).
\end{equation}  
The augmentations we added to the flux representation, hence, can then be interpreted as a smeared scalar field on $\realnumbers^3$. The smeared scalar field operator is of the form 
\begin{equation}
\sfield{g}=\INTset{\realnumbers^3}\dxd{3}{x}g(x)\sfield{x},
\end{equation}
with $\sfield{x}$ being an operator valued distribution and $g(x)$ a test function that is used to cast $\sfield{x}$ into a well-defined operator on the Fock space. 
Considering the pair $\pair{S}{f}$, we have to keep in mind that the Lie algebra valued smearing function only has support on the surface. Therefore we can work with $f$ actually only when we integrate it over the intrinsic parametrization of $S$. Furthermore we need an embedding of the surface into $\realnumbers^3$ in order to be able to integrate it against the scalar field. A possible definition is thus
\begin{equation}\label{eq_SUtwo_sfield}
\sfield{\Gamma\pair{S}{f}}=\INTset{\realnumbers^3}\dxd{3}{x}\Gamma\pair{S}{f}\sfield{x}=\INTset{\realnumbers^3}\dxd{3}{x}\INTset{S}\dxd{2}{u} \kernelS{_i}{x,u}f^i(u)\sfield{x}.
\end{equation} 
In this, $\pair{u_1}{u_2}$ is an intrinsic parametrization of $S$, and $\kernelS{_i}{x,u}$ is an integral kernel that relates points on the surface to points in $\realnumbers^3$. As an explicit example for $\Gamma(S,f)$ one might consider
\begin{equation}
\Gamma(S,f)=\Big(e^\Delta \sum_i f^i F_S\Big)(x),
\end{equation}
where $\Delta=\partial_i\partial^i$ is the negative-definite Laplacian on $\realnumbers^3$.

At the end of the day we have a representation of the holonomy flux algebra --~underlying loop quantum gravity~-- at disposal which is probably as close as possible to an actual almost quasifree representation of a theoretical Weyl algebra in the sense of definition \ref{def_WeylAlgebra}. 
The presumably most important characteristic of this representation is the fact that it is Gaussian for the fluxes in a nonextremal way. 
 
\subsection{Revisiting the area operator}\label{se:area}

In this section we consider the derivation of the area operator of loop quantum gravity for the representation introduced in the previous section.

In the context of the KS representation in \cite{Sahlmann:2010hn}, there is a similar augmentation to the flux representation of a classical flux. 
Also in \cite{Sahlmann:2010hn}, there is an analysis of the area operator in terms of the extended representation. So it is possible to compare the area operator for the almost quasifree representation to both the AL and the KS representation. 
It turns out that, in fact, the similarities to the representations in \cite{Sahlmann:2010hn} are sufficient to adopt their procedure and especially some crucial details to our situation. 
For reasons of simplicity, we do not consider the full augmented flux representation in \eqref{eq_SUtwo_FockRep}, but drop the classical flux part. Since this contribution to the AL representation is thoroughly dealt with in \cite{Sahlmann:2010hn}, we expect a similar behavior here. 

For the derivation we follow closely the steps in both, \cite{Ashtekar:1996eg,Sahlmann:2010hn}. Regarding the details of the original derivation of the area operator we refer to \cite{Ashtekar:1996eg}.
The area operator is a quantization of the classical area functional
\begin{equation}
A_S=\INTset{S}\dxd{2}{u}\sqrt{E_i(u)E^i(u)}. 
\end{equation}
In accordance with \cite{Ashtekar:1996eg} we want to use a  two-dimensional parametrization $u$ of the surface. The single component of the electric field we have to integrate over is $E_i^{\vphantom{3}} \equiv E_i^3$. 
The main issue of the derivation of the area operator is now to carefully regularize this integral and replace the classical expressions by their quantum counterparts.  
For a regularization one considers a family of non-negative densities $\regf{u,v}$, which, upon taking away the regularization parameter $\epsilon$, become Dirac deltas, {\ie}
\begin{equation}
\lim\limits_{\epsilon\rightarrow 0} \regf{u,w}=\diracdeltad{2}{u,w}.
\end{equation}
In the original derivation, this regularization allows for a point splitting of the area operator into a sum of operators that only acts at the intersection point of the graph with respect to which we want to determine the area of the surface. The same is considered here. 

In every representation at hand we do not have direct access to operator valued distributions for the electric field, but only to operators for fluxes. These can nevertheless be evaluated at certain points of the surface using a regularization as introduced above. One considers an $\SUtwo$ smearing function of the form 
$f\equiv \regf{u,w}\tau_i$ and determines the flux operators with respect to this. 
They can be denoted by 
\begin{equation}
\regflux{_i}{u}:=\fluxvfSfopen{S}{\regf{u,\slotdot}\tau_i},
\end{equation}
where the second argument of $\regf{u,w}$ is omitted because it is subject to some internal integration.\footnote{For the detailed form of the operators we refer to \cite{Ashtekar:1996eg} and the notation therein.} The surface is also omitted since the considerations refer to only a single surface.
The AL area operator then turns out to be the object 
\begin{equation}\label{eq_AreaOperatorALpart}
A_\AL(S)=\lim\limits_{\epsilon\rightarrow 0} \INTset{S}\dxd{2}{u}\sqrt{\regflux{_i}{u}\regflux{^i}{u}}= \sum\limits_{v\in V(\gamma)} A_{\AL,\gamma}(S),
\end{equation}
splitting into an expression of local area operators at the individual vertices of $\gamma$. Following \cite{Ashtekar:1996eg}, we raise and lower indices with $-\frac{1}{2}$ times the Cartan-Killing metric of $\sutwo$. The graph $\gamma$ is assumed to be adapted to the surface in the sense that there is a vertex at every intersection point. 

As shown in \cite{Ashtekar:1996eg}, the regularized fluxes act on a cylindrical function $\Psi_\gamma$  as
\begin{equation}
\regflux{_i}{u}\Psi_\gamma= 4\pi\lplanck^2\sum\limits_{v\in V(\gamma)}
\regf{x,v}\sum\limits_{J_v}\kappa_{J_v} X_{J_vi}\Psi_\gamma.
\end{equation}
Here $J_v$ is a label for the edges beginning or ending at vertex $v$ of $\gamma$ and $\lplanck=\sqrt{\hbar G}$. The operators $X_{J_vi}$ act for ingoing\,/\,outgoing edges as right-\,/\,left-invariant derivatives with respect to $e_{J_v}$ and $\kappa_{J_v}\equiv\keS{e_{J_v}}{S}$.
A straightforward calculation shows that at each vertex 
\begin{equation}\label{eq_CommFlux}
\commutator{X^i_{I_v}}{X^j_{J_v}}=i \delta_{I_vJ_v}\tensor{f}{^i^j_k}X^k_{J_v},
\end{equation}
where $\tensor{f}{^i^j_k}$ are the structure constants of $\sutwo$.
This implies that the left-\,/\,right-invariant vector fields $X^i_{J_v}$ individually satisfy --~at a fixed vertex and for each edge~-- the algebra relation of $\sutwo$ and hence can be treated like spin operators. Operators for different edges commute. 
The sum over the edges at a vertex splits, via the sign of $\kappa_{J_v}$, into a total contribution of type up and total contribution of type down.  

We now have to evaluate the scalar field operator at exactly the same smearing function. Equation \eqref{eq_SUtwo_sfield} therefore turns into  
\begin{equation}
\regsfield{_i}{u}:=\sfield{\Gamma\pair{S}{\regf{u,\slotdot}\tau_i}}=
\INTset{\realnumbers^3}\dxd{3}{x}\INTset{S}\dxd{2}{w} \kernelS{_i}{x,w}\regf{u,w} \sfield{x}.
\end{equation}
The full augmented representation of fluxes hence is 
\begin{equation}
\regFflux{_i}{u}:=\regflux{_i}{u}\otimes\unitelementof{\curlyF}+\unitelementof{\AL}\otimes\regsfield{_i}{u}.
\end{equation}
These are the objects we want to rederive the area operator for. 

Analogously to \cite{Ashtekar:1996eg,Sahlmann:2010hn} we consider the area operator as the limit of the regularized object 
\begin{equation}\label{eq_AreaOpRegularized}
A_{S,\epsilon}=\INTset{S}\dxd{2}{u}\sqrt{\regFflux{_i}{u}\regFflux{^i}{u}}.
\end{equation}
We analyze the object under the square root:
\begin{align}
g_{S,\epsilon}(u):=&\regFflux{_i}{u}\regFflux{^i}{u}\nonumber \\
=& \regflux{_i}{u}\regflux{^i}{u}\otimes\unitelementof{\curlyF}		\label{eq_gSeps_AL}\\
&+ \unitelementof{\AL}\otimes\regsfield{_i}{u}\regsfield{^i}{u}		\label{eq_gSeps_sfield}\\
&+ 2 \regflux{^i}{u}\otimes\regsfield{_i}{u}.		\label{eq_gSeps_mix}
\end{align}
The first term \eqref{eq_gSeps_AL} is exactly the term leading to the  AL area operator, while the second term \eqref{eq_gSeps_sfield} is of the same form but  depends only on the scalar field. Only in the third contribution \eqref{eq_gSeps_mix} is there an interaction between the AL and the scalar field parts of the flux representation. 

Following the ideas of \cite{Sahlmann:2010hn}, we realize that the individual contributions \eqref{eq_gSeps_AL}, \eqref{eq_gSeps_sfield}, and \eqref{eq_gSeps_mix}
mutually commute. This is obvious for considering the first and the second contribution. For the second and the third contribution it follows from the fact that $\commutator{\sfield{x}}{\sfield{y}}=0$. 
Finally, for the first and the third contribution the vanishing commutator follows from \eqref{eq_CommFlux}. As a consequence, there is a complete set of states that are eigenstates of all three terms. Having established this, we can now analyze the contributions individually.  

\textit{AL contribution:} We consider an arbitrary graph $\gamma$. Without loss of generality, we can refine it such that all transversal intersections of $\gamma$ with $S$ are vertices, and denote the result again with $\gamma$. We chose the parameter $\epsilon$ small enough such that for all $v,v'\in V(\gamma)$ the smearing function $f_\epsilon(v,v')$ is nonzero only for $v=v'$. Then the result of the analysis in \cite{Ashtekar:1996eg} for the piece $\regflux{_i}{u}\regflux{^i}{u}$ when acting on an eigenstate $\Psi_\gamma$ with an underlying spin network $\gamma$ is 
\begin{align}
&\regflux{_i}{u}\regflux{^i}{u}\Psi_\gamma= \sum\limits_{v\in V(\gamma)} \roundbrackets{\regf{u,v}}^2 \roundbrackets{a_v}^2 \Psi_\gamma,\\
a_v&=4\pi\lplanck^2\sqrt{2j^\mathrm{u}_{v}\roundbrackets{j^\mathrm{u}_{v}+1}+2j^\mathrm{d}_{v}\roundbrackets{j^\mathrm{d}_{v}+1}-j^\mathrm{u+d}_{v}\roundbrackets{j^\mathrm{u+d}_{v}+1}}.
\end{align}
The form of the eigenvalue $a_v$ comes from the fact that the eigenstates couple all spins of edges of type \emph{up} to a total spin $j^\mathrm{u}_{v}$ and all spins of edges of type \emph{down} to a total spin $j^\mathrm{d}_{v}$. The third contribution arises from coupling the up and down contributions. 

\textit{Scalar field contribution:}
The treatment of this is a little tricky. 
In the first place we are dealing with operator valued distributions integrated against test functions and need to consider eigenstates of such objects. 
As the operator valued distribution $\sfield{x}$ is the quantum field theoretical analog of the position operator in quantum mechanics it is clear that there are no proper eigenstates of $\sfield{x}$. One consequently has to go over to a formulation in terms of generalized eigenstates. 
For quantum mechanics this can even be formulated mathematically precise when considering the framework of rigged Hilbert spaces. Here one works with a triple $\curlyS\subset\curlyH\subset\curlyS^\prime$, which consists of the actual Hilbert space $H=\Ltwo{\realnumbers^3,\grad{^3x}}$, the space $\curlyS$ of test functions on $\realnumbers^3$ and its dual $\curlyS\pr$, the space of tempered distributions.  
The idea is now to transfer this procedure to the scalar field on a Fock space in order to be able to work with states that satisfy an eigenvalue equation of the form
\begin{equation}\label{eq_genFockES}
\sfield{f}|F)=F(f)|F),
\end{equation}
with the generalized eigenstate with respect to $F\in\curlyS\pr$ denoted by $|F)$ and $F(f)$ interpreted in a distributional sense, {\ie}
\begin{equation}
F(f)=\INTset{\realnumbers^3}\dxd{3}{x}F(x)f(x).
\end{equation}
A technical argument which we present in the Appendix, shows that it is reasonable to assume that there is a sufficiently large class of real valued functions $F$ that give rise to a tempered distribution and allow one to span the Hilbert space $\curlyH$. 
In a distributional sense Eq. \eqref{eq_genFockES} can be seen as
\begin{equation}
\sfield{x}|F)=F(x)|F).
\end{equation} 

With this at hand we can start to analyze the piece $\regsfield{_i}{u}\regsfield{^i}{u}$ by acting with a single field on a generalized eigenstate:
\begin{equation}
\begin{aligned}
\regsfield{_i}{u}|F)&=\INTset{\realnumbers^3}\dxd{3}{x}\INTset{S}\dxd{2}{w} \kernelS{_i}{x,w}\regf{u,w} \sfield{x}|F)\\
&=\INTset{\realnumbers^3}\dxd{3}{x}\INTset{S}\dxd{2}{w} \kernelS{_i}{x,w}\regf{u,w}F(x)|F)\\
&=:\reggenEV{_i}{u}|F).
\end{aligned}
\end{equation}
The whole contribution then is of course 
\begin{equation}
\regsfield{_i}{u}\regsfield{^i}{u}|F)=\reggenEV{_i}{u}\reggenEV{^i}{u}|F).
\end{equation}

\textit{Mixed contribution:}
For the final piece, we have to consider a state $\Psi_\gamma\otimes|F)$ where $\Psi_\gamma$ is an eigenstate of area at the vertex, and $|F)$ is an eigenstate of $\regsfield{_i}{u}$. Acting on this yields
\begin{equation}
\begin{aligned}
&\roundbrackets{2\regflux{^i}{u}\otimes\regsfield{_i}{u}}\Psi_\gamma\otimes|F)
=\Bigroundbrackets{8 \pi\lplanck^2 \sum\limits_{v\in V(\gamma)}
\regf{u,v} \roundbrackets{{X^{\mathrm{u}\,i}_v}-{X^{\mathrm{d}\,i}_v}}\otimes \reggenEV{_i}{u} }\Psi_\gamma\otimes|F), 
\end{aligned}
\end{equation} 
where we already split up the contributions for edges of type up and down. Note that in our choice of $\Psi_\gamma$ we have already specified the total spins $j^\mathrm{u}, j^\mathrm{d},j^\mathrm{u+d}$, but not the eigenvalue of one of their components. Now we will assume in addition that 
\begin{equation}
    X_v^{\mathrm{u}\,i}\, \reggenEV{_i}{u}\, \Psi_\gamma
    \equiv X_v^{\mathrm{u}\,i}\, \widehat{F}_{\epsilon i}(u)\, |F_\epsilon|(u)\,\Psi_\gamma
    = m_v^\mathrm{u}\, |F_\epsilon|(u)\,\Psi_\gamma,
\end{equation}
where $|F_\epsilon|(u)=\sqrt{\reggenEV{_i}{u}\reggenEV{^i}{u}}$ and $\widehat{F}_{\epsilon i}(u)= \reggenEV{_i}{u}/|F_\epsilon|(u)$. $m_v^\mathrm{u}$ is the eigenvalue of $X_v^{\mathrm{u}} \cdot \widehat{F}_{\epsilon}(u)$, i.e. the magnetic quantum number in the direction $\widehat{F}_{\epsilon}(u)$. We make the same assumption for $X_v^\mathrm{d}$, with the eigenvalue $m_v^\mathrm{d}$. Then
\begin{equation}
\roundbrackets{\roundbrackets{X^{\mathrm{u}\,i}_v-X^{\mathrm{d}\,i}_v}\otimes \reggenEV{_i}{u}} \Psi_\gamma\otimes|F)
=\roundbrackets{m_v\sqrt{\reggenEV{_i}{u}\reggenEV{^i}{u}}} \Psi_\gamma\otimes|F)
\end{equation}
and $m_v=m_v^\mathrm{u}-m_v^\mathrm{d}$ denotes the difference of total magnetic quantum numbers of up and down contributions. In total we have
\begin{equation}
\begin{aligned}
\roundbrackets{2\regflux{^i}{u}\otimes\regsfield{_i}{u}}\Psi_\gamma\otimes|F)
=\Bigroundbrackets{8\pi\lplanck^2 \sum\limits_{v\in V(\gamma)}
\regf{u,v} m_v\sqrt{\reggenEV{_i}{u}\reggenEV{^i}{u}}} \Psi_\gamma\otimes|F).
\end{aligned}
\end{equation}

Putting everything together, we find the following as the action of $g_{S,\epsilon}$ on the common eigenvector: 
\begin{equation}
\begin{aligned}
&g_{S,\epsilon}(u)\Psi_\gamma\otimes|F)=\\
&=\Bigroundbrackets{
\sum\limits_{v\in V(\gamma)} \roundbrackets{\regf{u,v}}^2 \roundbrackets{a_v}^2
+\reggenEV{_i}{u}\reggenEV{^i}{u}
+8\pi\lplanck^2  \sqrt{\reggenEV{_i}{u}\reggenEV{^i}{u}}  \sum\limits_{v\in V(\gamma)}
\regf{u,v} m_v
}\Psi_\gamma\otimes|F). 
\end{aligned}
\end{equation}
Following \cite{Sahlmann:2010hn}, we want to perform a completion of the square with respect to the first two terms in the bracket. That is, we will write the first two terms as a square minus a correction.\footnote{Since we combine two quadratic terms $a^2$, $b^2$, we can choose whether we want to obtain $(a+b)^2$ or $(a-b)^2$ by adding or subtracting a correction term, respectively. For the argument here we chose the first case. The second case yields the same result, however with an argument that is slightly more complicated.} 
To this end we have to have a look at the first term. At first we want to calculate the following:
\begin{equation}
\Bigroundbrackets{
\sum\limits_{v\in V(\gamma)} \regf{u,v} a_v
}^2
=
\sum\limits_{v,v\pr\in V(\gamma)} \regf{u,v}\regf{u,v\pr} a_v a_{v\pr}
=
\sum\limits_{v\in V(\gamma)} \roundbrackets{\regf{u,v}}^2 \roundbrackets{a_v}^2.
\end{equation}
This is in fact possible since we can choose $\epsilon$ to be so small that 
$\regf{u,v}\regf{u,v\pr}=0$ if the vertices do not coincide. 
Hence we rewrite 
\begin{equation}\label{eq_gSepsEV}
\begin{aligned}
g_{S,\epsilon}(u)\Psi_\gamma\otimes|F)=
\begin{aligned}[t]
\bigg(&\bigg(\sum\limits_{v\in V(\gamma)} \regf{u,v} a_v+ \sqrt{\reggenEV{_i}{u}\reggenEV{^i}{u}}\bigg)^2\\
&+\sum\limits_{v\in V(\gamma)} \regf{u,v} \sqrt{\reggenEV{_i}{u}\reggenEV{^i}{u}}
\roundbrackets{8\pi\lplanck^2 m_v-2 a_v}\bigg)\Psi_\gamma\otimes|F)
\end{aligned}
\end{aligned}
\end{equation}
and have found the eigenvalue of $g_{S,\epsilon}(u)$. 

However, the above eigenvalue is not the desired result. For this we have to take the square root and integrate over the surface. 
Again, we follow the descriptions in \cite{Sahlmann:2010hn}. We consider two real, positive variables $a\geq b$. As a matter of fact the variables satisfy the inequalities 
\begin{equation}
\begin{aligned}
\sqrt{a} &\leq \sqrt{a+b} \leq \sqrt{a}+\sqrt{b},\\
\sqrt{a}-\sqrt{b}&\leq \sqrt{a-b} \leq \sqrt{a}.
\end{aligned}
\end{equation}
The quadratic part of \eqref{eq_gSepsEV}, which we identify with $a$, is manifestly positive, and the absolute value of the linear part, which we identify with $b$, has to be smaller compared to the quadratic part in order to ensure positivity of the operator $g_{S,\epsilon}$. Hence, the inequalities are employable: the first one in the case where the linear part of \eqref{eq_gSepsEV} is positive, and the second in the case where it is negative.
The advantage of this is that $a$ and $b$ depend on the regulator $\epsilon$. 
If we can show that 
\begin{equation}\label{eq_SquarerootConditions}
\lim\limits_{\epsilon\rightarrow 0} \int \sqrt{b}=0,\quad \text{and}\quad \lim\limits_{\epsilon\rightarrow 0} \int \sqrt{a} \quad \text{exists},
\end{equation}
it holds furthermore, by means of the inequalities, that 
\begin{equation}
\lim\limits_{\epsilon\rightarrow 0} \int \sqrt{a\pm b}=\lim\limits_{\epsilon\rightarrow 0} \int \sqrt{a}.
\end{equation} 

Before we consider the square root of the  second term in \eqref{eq_gSepsEV} we take a closer look at 
\begin{equation}
\begin{aligned}
\reggenEV{_i}{u}\reggenEV{^i}{u}&=
\INTset{\realnumbers^3}\dxd{3}{x}\INTset{\realnumbers^3}\dxd{3}{y}
\INTset{S}\dxd{2}{v}\INTset{S}\dxd{2}{w}
\kernelS{_i}{x,v}\kernelS{^i}{y,w}
\regf{u,v}F(x)\regf{u,w}F(y)\\
\overset{\epsilon \rightarrow 0}{\longrightarrow}&
\INTset{\realnumbers^3}\dxd{3}{x}\INTset{\realnumbers^3}\dxd{3}{y}
\INTset{S}\dxd{2}{v}\INTset{S}\dxd{2}{w}
\kernelS{_i}{x,v}\kernelS{^i}{y,w}
\diracdeltad{2}{u,v}F(x)\diracdeltad{2}{u,w}F(y)=\\
&=
\INTset{\realnumbers^3}\dxd{3}{x}\INTset{\realnumbers^3}\dxd{3}{y}
\kernelS{_i}{x,u}\kernelS{^i}{y,u}
F(x)F(y)=:F_\phi(u)
\end{aligned}
\end{equation} 
and realize that in the limit where we remove the regulator this and hence also its square root are bounded functions on the surface, by means of the integral kernel and the generalized eigenvalue being test functions. 
Due to this we want to remove the regulator for this object now.

The final thing to realize is that if $\regf{u}$ is a density that converges to the delta function, its square root converges to zero. 
Now we are able to encounter the considered square root. We already take along the limit and the integration:
\begin{equation}
\begin{aligned}
&\lim\limits_{\epsilon\rightarrow 0}\INTset{S}\dxd{2}{u}\sqrt{
\left|
\sum\limits_{v\in V(\gamma)} \regf{u,v} \sqrt{F_\phi(u)}
\roundbrackets{8\pi\lplanck^2 m_v-2 a_v}\right|}=\\
&=\lim\limits_{\epsilon\rightarrow 0}\INTset{S}\dxd{2}{u}
\sum\limits_{v\in V(\gamma)} \sqrt{ \regf{u,v}} \sqrt[4]{F_\phi(u)}
\sqrt{|8\pi\lplanck^2 m_v-2   a_v|}\\
&=0
\end{aligned}
\end{equation}
since we can choose $\epsilon$ to be small enough to perform the summation outside of the square root and the absolute value.

According to this result, we finally can consider taking the square root of $g_{S,\epsilon}$, perform the integration, and remove the regulator:
\begin{equation}
\begin{aligned}
\lim\limits_{\epsilon\rightarrow 0}\INTset{S}\dxd{2}{u} \sqrt{g_{S,\epsilon}}\Psi_\gamma\otimes|F) 
&=\lim\limits_{\epsilon\rightarrow 0}\INTset{S}\dxd{2}{u} 
\Big|\sum\limits_{v\in V(\gamma)} \regf{u,v} a_v+ \sqrt{\reggenEV{_i}{u}\reggenEV{^i}{u}}\Big|\Psi_\gamma\otimes|F).
\end{aligned}
\end{equation}  
Both terms in the absolute value are manifestly positive. 
Hence, 
\begin{equation}
\begin{aligned}
\lim\limits_{\epsilon\rightarrow 0}\INTset{S}\dxd{2}{u} \sqrt{g_{S,\epsilon}}\Psi_\gamma\otimes|F) 
&=\biggroundbrackets{\lim\limits_{\epsilon\rightarrow 0}\INTset{S}\dxd{2}{u} 
\sum\limits_{v\in V(\gamma)} \regf{u,v} a_v+ \lim\limits_{\epsilon\rightarrow 0}\INTset{S}\dxd{2}{u} \sqrt{\reggenEV{_i}{u}\reggenEV{^i}{u}}}\Psi_\gamma\otimes|F)\\
&=\biggroundbrackets{ 
\sum\limits_{v\in V(\gamma)} a_v+ \INTset{S}\dxd{2}{u}\sqrt{F_\phi(u)}}\Psi_\gamma\otimes|F).
\end{aligned}
\end{equation}  
In the limit the integral takes away the Dirac delta and the first term of the above equation is just the original area operator. Replacing the eigenvalues again by the corresponding operators, the area operator for this almost quasifree representation reads 
\begin{equation}
A(S)=A_\AL(S)\otimes\unitelementof{\curlyF}+\unitelementof{\AL}\otimes\INTset{S}\dxd{2}{u}\sqrt{\phi_i(K_S(u))\phi^i(K_S(u))}. 
\end{equation}
Here we furthermore use some notation for the scalar fields once the regulator is removed:
\begin{equation}
\phi_i(K_S(u))=\INTset{\realnumbers^3}\dxd{3}{y}\kernelS{_i}{y,u}\sfield{y}.
\end{equation}

This result is in fact somewhat similar to \cite{Sahlmann:2010hn}. There, the original area operator is extended by the classical area of the surface. Another similarity is that there are only additive changes that increase
the quantum area. 
In any case, the addition here is an operator which carries its own quantum fluctuations.

\section{Conclusions and outlook} \label{se:outlook}
In this work, we presented a new type of vacuum and the corresponding representations of the HF algebra $\hfalgebra$. The key feature is a Gaussian vacuum expectation value for fluxes, encoding spatial geometry, which is characterized by a condensate contribution, e.\,g. a background flux, and fluctuations determined by a covariance of surfaces, here determined by e.\,g. a scalar field. 

In the case of $G=\Uone$ we found a precise relation between the representations of the HF algebra and those of a particular Weyl algebra. On it, we introduced a new class of almost quasifree states, which behaves as the AL state for holonomies and cylindrical functions, while it is Gaussian for fluxes. We worked out the representation in two examples. 

For the HF algebra $\hfalgebra$ defined in Sec. \ref{se:hfdef}, in the case $D=3$ and $G=\SUtwo$, we introduced a new class of representation, with the behavior described above. In particular, there are nonvanishing contributions for the $n$-point correlation functions. For a concrete example, we demonstrated that this change of representation leads to a significant change in the area spectrum of surfaces. 

It might be tempting to interpret the presence of a scalar field in the fluxes of the new representation as a toy model for matter coupling, especially because the geometric correlation functions are determined by the scalar field part. However, including the scalar field into the flux operators would lead to nonvanishing commutators between fluxes and the momentum conjugate to the scalar field. As geometric and matter variables have to have trivial commutation relations, the scalar field we introduced cannot be interpreted as a physical matter field. Rather, it only serves to introduce the Gaussianity in the representation.  

How can we extend and apply and extend the results contained in this work? 
\begin{itemize}

    \item[(i)] One area of application for  the new states is the quantum origin of the primordial perturbations. The current observations of the CMB suggest that primordial perturbations of the spatial metric (and matter density) are described well by a Gaussian random field with a certain covariance. Thus the states that we describe might be well suited to describe the quantum geometry of the early universe. In the standard picture, the covariance of the fluctuations is determined by following an initial quantum state through inflation. The new states allow one to think about a quantum gravitational origin of the fluctuations.

    \item[(ii)] What is the entanglement entropy between subsystems in the new class of states, and how does it compare to that of the class of states in \cite{Bianchi:2016tmw,Baytas:2018wjd,Bianchi:2018fmq}? In the U(1) case one can try to apply the techniques developed in \cite{Bianchi:2019pvv} to answer this question. A preliminary analysis indicates that the entanglement entropy of simple subsystems (generated by a flux and an intersecting holonomy) with the rest of the degrees of freedom is infinite. But this question should be studied further. 

    \item[(iii)] Are there states that are Gaussian in both variables? This question still stands and should eventually be resolved, either positively or negatively. There are some indications that it is not possible to find such states. In a slightly different framework, there is indeed a no-go result \cite{Lanery:2014bca,Lanery:2015mxa}. Also, for the $\Uone$ theory, there are some indications that no representations of that type can be found \cite{Nekovar14}. On the other hand, if such states do exist, then the almost quasifree states of the present work might be a stepping stone to reach them. 

\end{itemize}

\begin{acknowledgments}
The authors thank Eugenio Bianchi, Alexander Stottmeister, and Thomas Thiemann for helpful discussions, and Christian Fleischhack for many helpful comments that improved the manuscript substantially. RS thanks the Friedrich-Alexander-Universit\"at Erlangen-N\"urnberg (FAU) and the Evangelisches Studienwerk Villigst for financial support.
\end{acknowledgments}

\appendix*

\section*{Appendix: Eigenstates of a quantum field}
\renewcommand{\thesection}{A}
\label{app_EigenstatesQuantumfield}

In this appendix we will show that under reasonable assumptions there is a sufficiently large class of real valued functions $F$ such that there are generalized eigenstates $|F)$ for a scalar field that fulfill
\begin{equation}
\sfield{x}|F)=F(x)|F).
\end{equation} 

A rigged Hilbert space is a Hilbert space $\mathcal{H}$, together with a dense, continuously embedded topological vector space $\Phi\subset\mathcal{H}$. As a consequence, $\mathcal{H}$ is contained in the topological dual of $\Phi$, $\mathcal{H}\subset \Phi'$. The improper eigenstates of self-adjoint operators can find their home in such duals.
\begin{theorem}[\cite{em,Rigged}]
\label{thm_1}
Let $\Phi, \mathcal{H}, \Phi'$ be as above, with the additional assumption that $\mathcal{H}$ is separable. Any self-adjoint operator $ A $ mapping $\Phi$ continuously (in the topology of $ \Phi $) onto itself possesses a complete system of generalized eigenfunctions $ (F_{\alpha}) $, i.e. elements $ F_{\alpha} \in \Phi' $ such that for any $ \phi \in \Phi $,
\begin{equation}
{F_{\alpha}}(A \phi) = \lambda_{\alpha} {F_{\alpha}}(\phi), \qquad \alpha \in \mathfrak{A},
\end{equation}
where the set of values of the function $ \alpha \mapsto \lambda_{\alpha} $, $ \alpha \in \mathfrak{A} $, is contained in the spectrum of $A$ and has full measure with respect to the spectral measure $ {\sigma_{f}}(\lambda) $, of any element $f \in \mathcal{H} $. The completeness of the system means that $ {F_{\alpha}}(\phi) \neq 0 $ for any $ \phi \in \Phi $, $ \phi \neq 0 $, for at least one $ \alpha \in \mathfrak{A} $.
\end{theorem}
We will now show how this could be applied to a quantum field. For definiteness, we work with a scalar field on Minkowski space. We write
\begin{align}
\phi(x,t)&=\int\frac{d^3p}{(2\pi)^3}\frac{1}{\sqrt{2\omega_p}} \left(\anni_pe^{ip\cdot x-\omega_p t}+\crea_pe^{-ip\cdot x+\omega_p t}\right)\\
&=:\frac{1}{\sqrt{2}} \left(\phi_-(x,t)+\phi_+(x,t)\right). 
\end{align}
$\omega_p$ are the eigenvalues
\begin{equation}
\omega_p= \sqrt{p^2+m^2}
\end{equation}
of the operator 
\begin{equation}
\label{eq_energy_plane_wave}
E=\sqrt{-\Delta+m^2}, 
\end{equation}
and $a_p$, $\crea _p$ are standard momentum space annihilation/creation operators with 
\begin{equation}
[\anni_p,\crea_q]=(2\pi)^3\delta(p,q). 
\end{equation}
In the following, we will set $t=0$ and drop the time argument from all the functions. To simplify notation, we also define the operator 
\begin{equation}
    D=E^{\frac{1}{4}}. 
\end{equation}
Consequently, 
\begin{equation}
\label{eq_aadagger}
\anni(x):=(D^2\phi_-)(x), \qquad \crea(x):=(D^2\phi_+)(x)
\end{equation}
are standard momentum space annihilation/creation operators with 
\begin{equation}
\label{eq_ann}
[a(f_1),\crea (f_2)]=\scpr{f_1}{f_2}_{h}\unitelement
\end{equation}
over the Fock space 
\begin{equation}
\label{eq:fock_minkowski}
\mathcal{H}=\mathcal{F}(\curlyh), \qquad \curlyh=\mathcal{L}^2(\mathbb{R}^3, d^3x).
\end{equation}
Thus, both $\mathcal{H}$ and $\mathcal{h}$ are separable. 
The operators $\phi(f):= \int f(x)\phi(x)\,d^3x$ for smooth, real valued functions $f$ of compact support are mutually commuting and self-adjoint on $\mathcal{H}$. Therefore they must have a common set of generalized eigenstates. In the following, we want to investigate such states $|F)$ with the property 
\begin{equation}
\label{eq_eigen}
\phi(f)|F)= F(f)|F) 
\end{equation}
for a suitable class of real valued functions $F$. 
It will be useful to work with the dense domain $\mathcal{D}\subset \mathcal{H}$
\begin{equation}
\label{eq_nice}
\mathcal{D}=\text{span} \{ \prod_{\text{finite}} \crea(f_i)\ket{0}\;|\; f_1, f_2, \ldots\in \mathcal{S}(\mathbb{R}^3) \},
\end{equation}
where $\mathcal{S}(\mathbb{R}^3)$ denotes the Schwartz functions. $\mathcal{D}$ is contained in the domain of the $\phi(f)$. 
Consider the following operator: 
\begin{align}
O(F)&={\pi}^{-\frac{1}{4}}  
e^{\frac{1}{2}\scpr{D F}{DF}_h}
e^{-\frac{1}{2}\scpr{D\phi_+-\sqrt{2}DF }{D\phi_+-\sqrt{2}DF}}\\
&={\pi}^{-\frac{1}{4}}  
e^{-\frac{1}{2}\scpr{DF}{DF}_h}
e^{-\frac{1}{2}\int (D^{-1}\crea (x))^2\, d^3x}
e^{\sqrt{2} \crea (F)}. 
\label{eq_O}
\end{align}
This definition requires, at minimum that $F$ is such that $DF \in \curlyh$. We have the following.
\begin{lemma}
Formally, i.e. without consideration of domains, 
\begin{equation}
\phi(f)O(F)=F(f)O(F)+\frac{1}{\sqrt{2}}O(F) \anni (D^{-2} f).
\label{eq_eigen2}
\end{equation}
\end{lemma}
\begin{proof}
One can do a direct calculation, but it is easier to realize that 
\begin{equation}
[\anni(f), \Gamma[\crea (\slotdot)]]=\int f(x) \frac{\delta \Gamma[\crea (\slotdot)]}{\delta \crea (x)}\, d^3x
\end{equation}
where $\Gamma[\slotdot]$ is a functional which we assume to be differentiable. Then, noting 
\begin{equation}
\frac{\delta}{\delta \crea (x)} O(F)= -(D^{-2}\crea (x)-\sqrt{2}F(x))O(F),   
\label{eq_funct_deri_comm}
\end{equation}
one finds
\begin{equation}
\begin{aligned}
    \phi(f)O(F)&=\frac{1}{\sqrt{2}}(\anni(D^{-2}f)O(F)+\crea(D^{-2}f)O(F))\\
    &= \frac{1}{\sqrt{2}}([\anni(D^{-2}f),O(F)]+O(F)\anni(D^{-2}f)+O(F)\crea(D^{-2}f))\\
    &=  \frac{1}{\sqrt{2}}( -\crea (D^{-2}f)O(F)+\sqrt{2}F(f)O(F)+O(F)\anni(D^{-2}f)+O(F)\crea(D^{-2}f))\\
    &= F(f)O(F)+ \frac{1}{\sqrt{2}}O(F)\anni(D^{-2}f)
\end{aligned}
\end{equation}
as promised. For commuting $\crea(D^{-2}f)$ and $O(F)$ past each other we have appealed to the fact that $O(F)$ itself is defined entirely in terms of creation operators. 
\end{proof}
This lemma shows that $O(F)\ket{0}$ are formally the sought-for eigenstates \eqref{eq_eigen}. But these are obviously not normalizable, so in what sense do they even exist? 
\begin{lemma}
For $F\in \mathcal{S}'(\mathbb{R}^3),DF \in \curlyh $, the objects
\begin{equation}
|F)=O(F)\ket{0}
\end{equation}
define linear forms over the domain $\mathcal{D}$ \eqref{eq_nice}. 
\end{lemma}
\begin{proof}
We attempt to define the linear form
\begin{equation}
\langle X | F)= \sscpr{X}{O(F)}{0}, \qquad X\in \mathcal{D}
\end{equation}
by expanding the exponentials in a Taylor series and taking the limit. Since we assume $DF \in \curlyh$, the first exponential in \eqref{eq_O} is no problem, and since 
$\mathcal{D}$ only contains elements with finite particle number, the third exponential also represents no problem. But we have to consider the definition of the operator 
\begin{equation}
a_2=\int (D^{-1}\anni(x))^2\, d^3x
\end{equation}
and its adjoint which is used in the definition of $O(F)$. Using 
\begin{equation}
    a_2\, \crea(f)=2 \anni(D^{-2}f) +  \crea(f)\, a_2
\end{equation}
repeatedly to commute the annihilation operators to the right, one shows  
\begin{equation}
\scpr{\prod_k \crea (f_k)\Omega}{a_2^\dagger \Psi}\equiv
\scpr{a_2\prod_k \crea (f_k)\Omega}{\Psi}
= \sum_{(l,m), l \neq m} \scpr{f_l}{D^{-2} f_m}_h \scpr{\prod_{k\neq l,m} \crea (f_k)\Omega}{\Psi}. 
\end{equation}
Here $\Omega =\ket{0}$ the vacuum. Since Schwartz functions are also Schwartz after Fourier transform, it is easy to see that arbitrary positive and negative powers of $D$ leave $\mathcal{S}(\mathbb{R}^3)$ invariant. Consequently the inner products involving $D^{-2}$ are finite, and so are the products, and the sums. This shows that $a_2$ and its powers are well-defined on  $\mathcal{D}$ and, since $\mathcal{D}$ only contains elements with finite particle number, its exponential also represents no problem.
\end{proof}
Now note that the operators $\phi(f)$, $f\in \mathcal{S}(\mathbb{R}^3)$ map $\mathcal{D}$ into itself, because $D^{-1}$ maps  $\mathcal{S}(\mathbb{R}^3)$ into itself. We strongly suspect that $\mathcal{D}$ can be used to create a rigged Hilbert space, suitable for application of theorem \ref{thm_1}.
\begin{conjecture}
There is a topology on $\mathcal{D}$ that
\begin{enumerate}
\item is stronger than that induced from $\mathcal{H}$. 
\item is strong enough such that $|F)\in \mathcal{D}' \qquad \forall F: DF\in \curlyh$. 
\item is weak enough such that  $|F)\in \mathcal{D}' \qquad F: DF\in \curlyh$ comprise  all generalized eigenstates. 
\end{enumerate}
\end{conjecture}
If this conjecture is true, then 
\begin{corollary}
\begin{enumerate}
\item The joint eigenstates of $\phi(f)$, $f\in \mathcal{S}(\mathbb{R}^3)$ are 
\begin{equation}
|F)\quad \text{ with } F\in \mathcal{S}'(\mathbb{R}^3), \qquad DF\in L^2(\mathbb{R}^3)
\end{equation}
\item Functions of $\phi(f)$ have 
\begin{equation}
A(\phi(f)) |F)= A(F(f)) |F)
\end{equation}
\item If for two operators $A,B$ on $\Phi$
\begin{equation}
A |F) = B |F) \qquad \forall F:  F\in \mathcal{S}'(\mathbb{R}^3), \quad DF\in L^2(\mathbb{R}^3)
\end{equation} 
then $A=B$. 
\end{enumerate}
\end{corollary}

\end{document}